\RequirePackage{latexsym, amsmath, amstext, amssymb, amsfonts, amscd, bm, array, amsbsy, mathrsfs, amscd}
\documentclass[10pt, cmp, numbook, final, a4paper]{svjour}
\usepackage[10pt]{extsizes}
\pdfoutput=1
\usepackage{indentfirst}
\usepackage{graphics} 
\usepackage{epsfig}
\usepackage{transparent, colortbl}
\usepackage{booktabs}
\usepackage{multirow}
\usepackage{anysize}
\usepackage{latexsym}
\usepackage{pdfsync}
\usepackage[all]{xy}
\usepackage[usenames,dvipsnames]{xcolor}
\usepackage{enumerate}
\usepackage{float}
\restylefloat{table}
\usepackage{upgreek}
\usepackage{array}
\usepackage[colorlinks=true,linkcolor=black]{hyperref}
\colorlet{linkequation}{blue}

\newcommand*{\SavedEqref}{}
\let\SavedEqref\eqref
\renewcommand*{\eqref}[1]{%
  \begingroup
    \hypersetup{
      linkcolor=blue,
      linkbordercolor=blue,
    }%
    \SavedEqref{#1}%
  \endgroup
}

\usepackage{url}
\newcommand{\Hom}{{\rm Hom}}

\def\ad{\mathrm{ad}}

\newcommand{\Sym}{\mathrm{Sym}}

\newcommand{\bDelta}{\boldsymbol{\Delta}}

\newcommand{\bl}{\begin{Lemma}}
	\newcommand{\el}{\end{Lemma}}
\newcommand{\bt}{\begin{Theorem}}
	\newcommand{\et}{\end{Theorem}}
\newcommand{\bd}{\begin{definition}}
	\newcommand{\ed}{\end{definition}}
\newcommand{\End}{\mathrm{End}}
\newcommand{\Aut}{\mathrm{Aut}}

\newcommand{\eqdef}{\stackrel{{\rm def.}}{=}}

\DeclareFontFamily{U}{rsf}{}
\DeclareFontShape{U}{rsf}{m}{n}{<5> <6> rsfs5 <7> <8> <9> rsfs7 <10-> rsfs10}{}
\DeclareMathAlphabet\Scr{U}{rsf}{m}{n}

\def\Z{\mathbb{Z}}

\def\R{\mathbb{R}}
\def\bD{\mathbb{D}}

\def\rk{{\rm rk}}

\def\dd{\mathrm{d}}

\def\ad{\mathrm{ad}}

\def\bcD{\boldsymbol{\mathcal{D}}}

\def\bJ{\mathbf{J}}

\def\bomega{\boldsymbol{\omega}}

\def\bJ{\mathbf{J}}

\def\bt{\mathbf{t}}

\def\U{\mathrm{U}}

\def\graph{\mathrm{graph}}
\def\ungraph{\mathrm{ungraph}}
\def\bTheta{\mathbf{\Theta}}
\def\bPsi{\mathbf{\Psi}}

\newcommand{\be}{\begin{equation*}}
\newcommand{\ee}{\end{equation*}}
\newcommand{\ben}{\begin{equation}}
\newcommand{\een}{\end{equation}}
\newcommand{\beqa}{\begin{eqnarray*}}
	\newcommand{\eeqa}{\end{eqnarray*}}
\newcommand{\beqan}{\begin{eqnarray}}
\newcommand{\eeqan}{\end{eqnarray}}

\newcommand{\nn}{\nonumber}

\newcommand{\id}{\mathrm{id}}
\newcommand{\Tr}{\mathrm{Tr}}

\def\cC{{\mathcal C}}

\def\loslash{\bm{\oslash}}

\newcommand{\cM}{\mathcal{M}}

\def\Hol{\mathrm{Hol}}

\def\cD{\mathcal{D}}

\def\cE{\check{E}}
\def\cX{\mathcal{X}}

\def\cP{\mathcal{P}}

\def\cG{\mathcal{G}}

\def\cF{\mathcal{F}}
\def\cC{\mathcal{C}}
\def\cH{\mathcal{H}}
\def\bXi{{\boldsymbol{\Xi}}}

\def\G_2{\mathrm{G_2}}

\def\cS{\mathcal{S}}

\def\cV{\mathcal{V}}

\def\mf{\mathbf{f}}

\def\G{\mathrm{G}}

\def\cE{\mathcal{E}}
\def\Isom{\mathrm{Isom}}
\def\bg{\mathbf{g}}

\def\mT{\mathbf{T}}
\def\bPhi{\mathbf{\Phi}}
\def\bConf{\mathbf{Conf}}
\def\bSol{\mathbf{Sol}}

\def\cX{\mathcal{X}}

\def\grad{\mathrm{grad}}

\def\bD{\mathbf{D}}

\def\bd{\boldsymbol{\dd}}

\def\bg{\mathbf{g}}

\def\bQ{\boldsymbol{Q}}

\def\bomega{\boldsymbol{\omega}}

\def\bgamma{\boldsymbol{\gamma}}

\def\rS{\mathrm{S}}

\def\Met{\mathrm{Met}}
\def\Iso{\mathrm{Iso}}

\def\i{\mathbf{i}}
\def\sc{\mathrm{sc}}

\def\Sol{\mathrm{Sol}}
\def\Conf{\mathrm{Conf}}
\def\ub{\mathrm{ub}}

\def\fg{\mathfrak{g}}

\def\triv{\mathrm{triv}}
\def\mG{\mathbb{G}}
\def\bcS{\boldsymbol{\mathcal{S}}}
\def\bD{\mathbf{D}}
\def\mq{\mathbf{q}}

\def\bsigma{\boldsymbol{\sigma}}
\def\bcV{\boldsymbol{\cV}}

\def\cl{\mathrm{cl}}

\def\fh{\mathfrak{h}}

\def\tvarphi{{\tilde \varphi}}

\begin{document}

\title{Section sigma models coupled to symplectic duality bundles on Lorentzian four-manifolds}
  
\author{C. I. Lazaroiu\inst{1} and C. S. Shahbazi\inst{2}}
\institute{Center for Geometry and Physics, Institute for Basic
  Science, Pohang 790-784, Republic of Korea, \email{calin@ibs.re.kr}
  \and Department of Mathematics, University of Hamburg, Germany, \email{carlos.shahbazi@uni-hamburg.de}}

\date{}

\maketitle

%\thanks{2010 MSC. Primary:  53C80. Secondary: 53C50, 83C10.}
%\keywords{Supergravity, Lorentzian geometry, symplectic geometry}

\begin{abstract}
We give the global mathematical formulation of a class of generalized
four-dimensional theories of gravity coupled to scalar matter and to
Abelian gauge fields. In such theories, the scalar fields are
described by a section of a surjective pseudo-Riemannian submersion
$\pi$ over space-time, whose total space carries a Lorentzian metric
making the fibers into totally-geodesic connected Riemannian submanifolds.  In
particular, $\pi$ is a fiber bundle endowed with a complete Ehresmann
connection whose transport acts through isometries between the
fibers. In turn, the Abelian gauge fields are ``twisted'' by a flat
symplectic vector bundle defined over the total space of $\pi$. This
vector bundle is endowed with a vertical taming which locally encodes
the gauge couplings and theta angles of the theory and gives rise to the notion of twisted self-duality, of crucial importance to construct the theory. When the Ehresmann connection of $\pi$ is integrable, we show that our theories are locally equivalent to ordinary Einstein-Scalar-Maxwell theories and
hence provide a global non-trivial extension of the universal bosonic sector of
four-dimensional supergravity. In this case, we show using a special
trivializing atlas of $\pi$ that global solutions of such models can
be interpreted as classical ``locally-geometric'' U-folds. In the
non-integrable case, our theories differ locally from ordinary
Einstein-Scalar-Maxwell theories and may provide a geometric
description of classical U-folds which are ``locally non-geometric''.
\end{abstract}

\tableofcontents

\section{Introduction}

The construction of four-dimensional supergravity theories usually
found in the physics literature (see, for example,
\cite{Andrianopoli,AndrianopoliII,AndrianopoliIII,Aschieri}) is 
local in the sense that it is carried out ignoring the topology of the space-time
manifold and without specifying the precise global description of the
configuration space or the global mathematical structures required to
define it. Such constructions are discussed traditionally only on a
contractible subset $U$ of space-time, which guarantees that any fiber
bundle defined on $U$ is trivial and hence that any section of such a
bundle can be identified with a map from $U$ into the
fiber. Accordingly, the physics literature traditionally treats all
classical fields as functions defined on $U$ and valued in some target
space, which is either a vector space or (for the scalar fields) a
manifold $\cM$ endowed with a Riemannian metric $\cG$. It is often
also tacitly assumed that $\cM$ is contractible, which implies that
the duality structure \cite{GESM,GESMNote} of the Abelian gauge theory
coupled to the scalar fields is described by a trivial flat symplectic
vector bundle $\cS$ defined on $\cM$, whose data can then be encoded
by a symplectic vector space $\cS^0$ (the fiber of $\cS$)
\cite{AndrianopoliII,AndrianopoliIII}. Due to this assumption, the
gauge field strengths and their Lagrangian conjugates are usually
treated as two-forms defined on $U$ and valued in $\cS^0$. It is not
unambiguously clear how such local constructions can be made into
complete, mathematically rigorous, global definitions of classical
theories of matter coupled to gravity when $M$ and $\cM$ are not
contractible. To fully define such theories, one must decide how to
interpret globally various local formulas and differential operators.
Such global interpretations are generally non-unique in the sense
that they depend on choices of auxiliary geometric data
which are not visible in the usual local formulation
\cite{GESM,GESMNote,GeometricUfolds}.

In this paper, we consider the issue of finding a general global
mathematical formulation of the universal bosonic sector of
four-dimensional supergravity theories, which consists of gravity
coupled to scalars and to Abelian gauge fields. We find, among other
results, that the traditional local formulas
\cite{Andrianopoli,AndrianopoliII,Aschieri,GaillardZumino,Ortin,FreedmanProeyen}
can be interpreted globally in a manner which provides a mathematical
definition of ``classical locally-geometric U-folds'' as global
solutions of the equations of motion of the resulting globally-defined
classical theory. Furthermore, the global theory that we obtain
represents a non-trivial extension of the standard bosonic sector of
ungauged four-dimensional supergravity. The supersymmetrization of the
theories introduced in this paper is currently an open problem
involving several objects and structures of mathematical interest,
such as spinor bundles and Lipschitz structures, Special-K\"ahler and
Quaternionic-K\"ahler manifolds or exceptional Lie groups, all
interacting in a delicate equilibrium dictated by supersymmetry. 

In our construction, the scalar map of the sigma model is first
promoted to a section of a {\em Kaluza-Klein space}
defined over space-time and endowed with a {\em vertical scalar
  potential} $\bPhi$, leading to a {\em section sigma model} for the
scalar fields. The latter is described by a Lagrangian density defined
on the space of sections of $\pi$. In the particular case when $\bPhi$
vanishes, the solutions of the equations of motion of the section
sigma model are the {\em pseudoharmonic sections} studied by
C.~M.~Wood \cite{Wood1,Wood2,Wood3,Wood4,Wood5,Wood6}. The section
sigma model is then coupled to Abelian gauge fields governed by a
non-trivial duality structure, thereby extending the construction of
\cite{GESM} to this more general setting. The globally-defined theory
obtained in this manner will be called {\em generalized
  Einstein-Section Maxwell} (GESM) theory.

By definition, a {\em Kaluza-Klein space} over a Lorentzian
manifold $(M,g)$ is a surjective submersion $\pi:E\rightarrow M$ whose total
space $E$ is endowed with a Lorentzian metric $h$ (known as a {\em
  Kaluza-Klein metric}) such that $\pi$ is a pseudo-Riemannian
submersion \cite{ONeillBook} from $(E,h)$ to $(M,g)$, whose fibers are
totally geodesic Riemannian connected submanifolds of $(E,h)$. Such
spaces were studied in the literature on the mathematical foundations
of Kaluza-Klein theory (see, for example,
\cite{Bourguignon,Hogan,Betounes1,Betounes2}) and their Riemannian
counterpart played an important role in the construction of non-trivial
examples of Einstein metrics \cite{Besse} and in the study of metrics
of positive sectional curvature \cite{Ziller}. The orthogonal
complement of the vertical distribution of $\pi$ with respect to a
Kaluza-Klein metric $h$ is a horizontal distribution $H$ which gives a
complete Ehresmann connection for $\pi$. Since the fibers are totally
geodesic, the Ehresmann transport $T$ defined by $H$ proceeds through
isometries between the fibers. As a consequence, all fibers can be
identified with some model Riemannian manifold $(\cM,\cG)$ and the
holonomy group $G$ of the Ehresmann connection is a subgroup of the
isometry group $\Iso(\cM,\cG)$ of the fiber, thereby being a
finite-dimensional Lie group. This implies that $\pi$ is a fiber
bundle with structure group $G$ (a fiber $G$-bundle in the sense of
\cite{Kollar,Michor}) and hence is associated to a principal
$G$-bundle $\Pi$ (known as the {\em holonomy bundle}) through the
isometric action of $G$ on $\cM$. Moreover, the Ehresmann connection
$H$ is induced by a principal connection $\theta$ defined on $\Pi$. In
fact, the Kaluza-Klein metric $h$ is uniquely determined by $H$, $g$
and $\cG$ or, equivalently, by $\theta$, $g$ and $\cG$ together with
an embedding $G\hookrightarrow \Iso(\cM,\cG)$. We say that the
Kaluza-Klein space is {\em integrable} if the distribution $H$ is
Frobenius integrable, which happens if and only if the principal
connection $\theta$ is flat. A {\em vertical scalar potential} on
$\pi$ is a smooth real-valued map $\bPhi$ defined on the total space
$E$ of $\pi$ such that the restrictions of $\bPhi$ to the fibers of
$\pi$ are related by the Ehresmann transport and hence can be
identified with a smooth real-valued function $\Phi$ defined on $\cM$, 
which is invariant under the holonomy group $G$. 

When $\pi$ is a topologically trivial fiber bundle and $H$ is a
trivial Ehresmann connection, the sections of $\pi$ can be identified
with the graphs of smooth functions from $M$ to $\cM$ and the section
sigma model reduces to the ordinary scalar sigma model defined by the
``scalar structure'' $(\cM,\cG,\Phi)$ \cite{GESM,GESMNote}. This
reduction always happens when $M$ is contractible and $H$ is
integrable, since in that case $\pi$ is necessarily topologically
trivial and $H$ is necessarily a trivial Ehresmann connection. In
particular, scalar sigma models with integrable Kaluza-Klein space are
{\em locally} indistinguishable from ordinary scalar sigma
models. When $\pi$ is topologically trivial but $H$ is not integrable,
the section sigma model reduces to a {\em modified} nonlinear sigma
model for maps from $M$ to $\cM$, which, to our knowledge, has not
been considered before. As a consequence, section sigma models with
non-integrable Kaluza-Klein space are {\em locally} distinct from
ordinary sigma models.

As mentioned above, section sigma models can be coupled
to Abelian gauge fields, leading to the construction of generalized
Einstein-Section-Maxwell (GESM) theories. The most general coupling of the
ordinary scalar sigma model to Abelian gauge fields can be described
globally \cite{GESM} using a so-called ``electromagnetic structure''
defined on $\cM$. In our situation, this is promoted to a 
``horizontally constant'' electromagnetic structure defined on the
total space $E$ of the Kaluza-Klein space. By definition, this is a
tamed flat symplectic vector bundle $\bcS$ defined over $E$ such that
the taming is invariant under the lift of
the Ehresmann transport of $\pi$ along the flat connection of
$\bcS$.

The fiber bundle $\pi$ admits {\em special trivializing atlases}
supported on convex covers $(U_\alpha)_{\alpha\in I}$ of $M$. The
local trivialization maps of such an atlas are constructed using the
Ehresmann transport along geodesics lying inside $U_\alpha$ and ending
at some reference point chosen in each $U_\alpha$. The behavior of the
section sigma model with respect to such atlases depends on whether
the Ehresmann connection $H$ is integrable or not:
\begin{enumerate}[A.]
	
\itemsep 0.0em

\item When the Kaluza-Klein space is integrable, any special
  trivializing atlas allows one to identify the restriction of
  $(\pi,H)$ to $U_\alpha$ with the topological trivial fiber bundle
  $(E^0_{\alpha}\eqdef U_\alpha\times \cM, \pi^0_{\alpha}\eqdef
  \pi|_{E_\alpha^0})$ defined over $U_\alpha$, endowed with the
  trivial Ehresmann connection. In this case, the restriction of the
  section sigma model to $U_\alpha$ identifies with the ordinary sigma
  model of maps from $U_\alpha$ to $\cM$. This implies that a global
  solution of the section sigma model can be obtained by patching
  local solutions of the ordinary scalar sigma model using symmetries
  of the equations of motion of the later and thus it can be
  interpreted as a {\em classical} counterpart of a U-fold. Similar
  statements apply after coupling to Abelian gauge fields. Thus:
	
\
	
\noindent{\em When the Kaluza-Klein space is integrable, global
  solutions of the GESM theory correspond to classical
  locally-geometric U-folds glued from local solutions of the ordinary
  sigma model coupled to Abelian gauge fields (which may have a
  non-trivial duality structure) using symmetries of the equations of
  motion of the latter.}
	
\
	
\noindent This can be taken as a rigorous {\em definition} of
classical locally-geometric U-folds and may provide a global geometric
description of the classical limit of certain string theory U-folds in
four dimensions.
	
\

\item When the Kaluza-Klein space is not integrable, the
  restriction of $H$ to $U_\alpha$ identifies with a non-integrable
  horizontal distribution of $\pi^0_\alpha$. In this case, local
  sections of $\pi$ defined over $U_\alpha$ are the graphs of smooth
  maps into the fiber $\cM$ and local solutions can be interpreted as
  solutions of the {\em modified} non-linear sigma model for maps
  mentioned above. Hence:
	
\
	
\noindent When the Kaluza-Klein space is not integrable,
global solutions of the GESM theory can be glued from local solutions
of the modified sigma model coupled to Abelian gauge fields (which may
have a non-trivial duality structure).
	
\
	
\end{enumerate}
Together with the results of \cite{GESM}, point A. above implies that
GESM theories with integrable Kaluza-Klein space are {\em locally}
indistinguishable from usual Einstein-Scalar-Maxwell theories and
hence provide admissible global interpretations of the local formulas
governing the universal bosonic sector of four-dimensional
supergravity. On the other hand, point B. above implies that GESM
theories with non-integrable Kaluza-Klein space are {\em locally
  distinct} from ordinary ESM theories.

The paper is organized as follows. Section \ref{sec:KK} discusses
Kaluza-Klein spaces, their special trivializing atlases and their
classification in the integrable case. The same section discusses
vertical scalar potentials and bundles of scalar structures as well as
the classification of the latter in the integrable case.  Section
\ref{sec:scalarsection} discusses section sigma models and the U-fold
interpretation of their global solutions in the integrable
case. Section \ref{sec:scalarelectro} discusses bundles of
scalar-electromagnetic structures, which are Kaluza-Klein spaces
endowed with ``horizontally-constant'' data describing the scalar
potential and electromagnetic bundle needed to couple the section
sigma model to Abelian gauge fields. Section \ref{sec:GESM} gives the
global formulation of a generalized Einstein-Section-Maxwell
theory. In the integrable case, we show that special trivializing
atlases of the underlying Kaluza-Klein space allow one to interpret
global solutions of such models as classical locally-geometric
U-folds. This section contains the main result of the paper, namely
Theorem \ref{thm:equivalence}, which proves in precise terms the {\em
  local} equivalence between GESM theories and ordinary ESM
theories. Section \ref{sec:SS} illustrates section sigma models with a
simple example, showing how in a special case they recover the
celebrated Scherk-Schwarz construction\footnote{Albeit {\em without}
  gauging of any putative continuous isometry of the scalar
  manifold. In fact, our construction does {\em not} assume existence
  of any continuous isometries.}. Finally, Section
\ref{sec:conclusions} contains our conclusions and some directions for
further research. Appendix \ref{app:submersions} contains technical
material on pseudo-Riemannian submersions and Kaluza-Klein
spaces. Appendix \ref{app:localadapted} shows that, in the
non-integrable case, a section sigma model reduces locally to a
modified sigma model for maps, which we describe explicitly using
adapted local coordinates. In addition, the same Appendix gives local
expressions in adapted coordinates for some key objects used in the
formulation of GESM theories.

\subsection{Notations and conventions}

All manifolds considered are smooth, connected, Hausdorff and
paracompact (hence also second countable) while all fiber bundles
considered are smooth. All submersions considered are assumed to be
surjective and to have connected fibers.  Given vector bundles $\cS$
and $\cS'$ over some manifold $M$, we denote by $Hom(\cS,\cS')$,
$Isom(\cS,\cS')$ the bundles of morphisms and isomorphisms from $\cS$
to $\cS'$ and by $\Hom(\cS,\cS')\eqdef \Gamma(M, Hom(\cS,\cS'))$,
$\Isom(\cS,\cS')\eqdef \Gamma(M, Isom(\cS,\cS'))$ the sets of smooth
sections of these bundles. When $\cS'=\cS$, we set $End(\cS)\eqdef
Hom(\cS,\cS)$, $Aut(\cS)\eqdef Isom(\cS,\cS)$ and $\End(\cS)\eqdef
\Hom(\cS,\cS)$, $\Aut(\cS)\eqdef \Isom(\cS,\cS)$. Given a smooth map
$f:M_1\rightarrow M_2$ and a vector bundle $\cS$ on $M_2$, we denote
the $f$-pull-back of $\cS$ to $M_1$ by $\cS^f$. Given a section $s\in
\Gamma(M_2,\cS)$, we denote its $f$-pullback by $s^f\in
\Gamma(M_1,\cS^f)$. Given $T\in \Hom(\cS,\cS')$, where $\cS,\cS'$ are
vector bundles defined on $M_2$, we denote the $f$-pullback of $T$ by
$T^f\in \Hom(\cS^f,(\cS')^f)$.  Let $\Met_{p,q}(W)$ denote the set of
non-degenerate symmetric pairings of signature $(p,q)$ on a vector
bundle $W$ of rank $\rk W=p+q$. Let $\Met(W)\eqdef \sqcup_{p,q\geq 0,
  p+q=\rk W}\Met_{p,q}(W)$ denote the set of all non-degenerate
metrics on $W$. When $W=TM$ is the tangent bundle of a manifold $M$,
we set $\Met_{p,q}(M)\eqdef \Met_{p,q}(TM)$ and $\Met(M)\eqdef
\Met(TM)$.  By definition, a Lorentzian manifold is a
pseudo-Riemannian manifold of ``mostly plus'' signature. Given a
manifold $M$, let $\cP(M)$ denote the set of paths (piecewise-smooth
curves) $\gamma:[0,1]\rightarrow M$. We sometimes use various
notations and conventions introduced in \cite{GESM}.

\begin{remark}
Throughout the paper, we use the mathematical concept of
``Kaluza-Klein space'', which is well-established in the mathematics
and mathematical physics literature (see, for example,
\cite{Bourguignon,Hogan,Betounes1,Betounes2,Besse}). Such spaces were
initially defined and studied in the context of mathematical
foundations of Kaluza-Klein theories, which arise by reducing a
higher-dimensional theory of gravity on such a manifold. In the present
paper, such spaces are used merely as convenient auxiliary
mathematical data which parameterize a GESM theory, despite the
historical context in which they were introduced initially. Throughout
the paper, no reduction of any putative higher dimensional theory on
such spaces is ever suggested or performed.
\end{remark}

\section{Kaluza-Klein spaces, vertical potentials and bundles of scalar structures}
\label{sec:KK}

\subsection{Lorentzian submersions and Kaluza-Klein metrics} 

Let $(M,g)$ be a connected four-dimensional Lorentzian manifold. Let
$E$ be a connected manifold of dimension $\dim E=n+4$, where $n>0$. 
Recall that a smooth map $\pi:E\rightarrow M$ is called a {\em
  surjective submersion} if $\pi$ is surjective and the linear map
$\dd_e \pi\colon T_e E\rightarrow T_{\pi(e)} M$ is surjective for all
$e\in E$. We shall always assume that the fibers of $\pi$ are
connected.

Let $\pi:E\rightarrow M$ be a surjective submersion with connected fibers.  The
{\em vertical distribution} of $\pi$ is the rank $n$ distribution
$V\eqdef \ker (\dd \pi)\subset T E$ defined on $E$, which is Frobenius
integrable and integrates to the foliation whose leaves are the fibers
of $\pi$. A Lorentzian metric $h$ on $E$ is called {\em
  $\pi$-positive} if its restriction to each fiber of $\pi$ is
positive-definite. In this case, the $h$-orthogonal complement
$H:=H(h)$ of $V$ inside $TE$ is a distribution of rank four called the
{\em horizontal distribution} of $\pi$ determined by $h$. Let $h_V$
and $h_H$ denote the metrics induced by $h$ on $V$ and $H$, which we
call the {\em vertical} and {\em horizontal} metrics induced by
$h$. For any $e\in E$, the map $\dd_e \pi$ restricts to a linear
bijection from $H_e$ to $T_{\pi(e)}M$. Hence the restriction
$\dd\pi|_H$ of $\dd \pi$ to $H$ induces a based isomorphism of vector
bundles $(\dd\pi)_H\colon H\stackrel{\sim}{\rightarrow} (TM)^\pi$.

\begin{definition}
Let $h$ be a $\pi$-positive Lorentzian metric on $E$. 
The surjective submersion $\pi:E\rightarrow M$ is called a {\em
  Lorentzian submersion} from $(E,h)$ to $(M,g)$ if the bundle
isomorphism $(\dd\pi)_H$ is an isometry from $(H,h_H)$ to
$((TM)^\pi,g^\pi)$. 
\end{definition}

\noindent 
When $\pi$ is a Lorentzian submersion from $(E,h)$ to $(M,g)$, the
pair $(H,h_H)$ is a pseudo-Euclidean distribution of signature $(3,1)$
defined on $E$ while $(V,h_V)$ is a Euclidean distribution. In
particular, $\pi$ is a surjective pseudo-Riemannian submersion in the
sense of \cite{ONeillBook}.

\begin{definition}
A {\em Kaluza-Klein metric} for the surjective submersion
$\pi:E\rightarrow M$ relative to the Lorentzian metric $g\in
\Met_{3,1}(M)$ is a $\pi$-positive Lorentzian metric $h$ on $E$
which makes $\pi$ into a Lorentzian submersion from $(E,h)$ to
$(M,g)$.
\end{definition}

\noindent More information about pseudo-Riemannian submersions and
Kaluza-Klein metrics can be found in Appendix \ref{app:submersions}.

Any horizontal distribution $H$ defines a horizontal lift of vector
fields, which takes $Q\in \cX(M)$ into the unique horizontal vector
field $\bar{Q}\in \Gamma(E,H)$ satisfying $\dd\pi(\bar{Q}) = Q$. Moreover, 
any vector field $X\in \cX(E)$ decomposes uniquely as $X=X_H\oplus X_V$, 
where $X_H\in \Gamma(E,H)$ and $X_V\in \Gamma(E,V)$. The
{\em curvature} $\cF\in \Omega^2(E,V)$ of $H$ is defined
as \cite{Michor}:
\be
\cF(X,Y)\eqdef [X_H,Y_H]_V\, , \quad \forall X,Y\in \cX(E)\, .
\ee
Its restriction to $H$ gives a section
$\cF_H\in \Gamma(E,\wedge^2H^\ast\otimes V)$ which satisfies:
\be
\cF_H(X,Y)=\cF(X,Y)=[X,Y]_V\, , \quad \forall X,Y\in \Gamma(E,H)
\ee
and hence coincides up to a constant factor with the restriction to
$H$ of O'Neill's second fundamental tensor $A$ \cite{ONeill} of the
pseudo-Riemannian submersion $\pi$:
\be
A_XY=-A_YX=\frac{1}{2}\cF_H(X,Y)\, , \quad \forall X,Y\in \Gamma(E,H)~~.
\ee
A basic property of O'Neill's tensor is that it vanishes if and only
if $\cF_H$ does. Hence either of $A$ or $\cF_H$ describe the
obstruction to Frobenius integrability of $H$.

The distribution $H$ is called {\em complete} if the flow $T$ defined
by horizontal lifts of vector fields is globally-defined, which
amounts to the condition that {\em any} curve in $M$ lifts to a
horizontal curve in $E$ through any point lying in the fiber above its
source.  When $H$ is complete, it follows from a result of
\cite{Ehresmann} that $\pi$ is a fiber bundle, though its structure
group need not be a finite-dimensional Lie group. In this case, $H$ is
an Ehresmann connection for $\pi$ and $T$ is called its Ehresmann
transport. Reference \cite{Reckziegel} shows that a sufficient
condition for $H$ to be complete and for the structure group to be a
finite-dimensional Lie group is that the fibers of $\pi$ be
geodesically complete and totally geodesic connected submanifolds of
$(E,h)$. In this case the Ehresmann transport $T$ is isometric, which
means that $T_\gamma$ is an isometry from $E_{\gamma(0)}$ to
$E_{\gamma(1)}$ for any path $\gamma\in \cP(M)$. Hence the structure
group is isomorphic to the Riemannian isometry group of the fiber of
$\pi$ over any given point in $M$. 

\subsection{Lorentzian Kaluza-Klein spaces}

\begin{definition}
A {\em Lorentzian Kaluza-Klein space} over $(M,g)$ is a Lorentzian
submersion $\pi:(E,h)\rightarrow (M,g)$ such that $(E,h)$ is connected
and such that the fibers of $\pi$ are geodesically complete and
totally geodesic connected submanifolds of $(E,h)$. The Kaluza-Klein
space is called {\em integrable} if the horizontal distribution
$H(h)\subset TE$ defined by $h$ is Frobenius integrable.
\end{definition}

\noindent
Consider a Lorentzian Kaluza-Klein space $\pi:(E,h)\rightarrow (M,g)$ with
horizontal distribution $H:=H(h)$, whose Ehresmann transport we denote
by $T:=T(h)$. Let $m_0$ be fixed point of $M$ and let $\cM\eqdef
E_{m_0}=\pi^{-1}(\{m_0\})$ and $\cG\eqdef h|_{E_{m_0}}$. As mentioned
above, the results of \cite{Reckziegel} imply that $H$ is a complete
Ehresmann connection, that $\pi$ is a fiber bundle and that the
transport $T$ is isometric. Hence the restriction $h_m\eqdef h|_{E_m}$
of $h$ to the fibers of $E$ is uniquely determined by $\cG$ and by
$T$. Therefore, the metric $h$ of a Kaluza-Klein space is uniquely
determined by $g$, $\cG$ and $H$. When $g$ and $\cG$ are fixed, we
thus have a bijection between Ehresmann connections $H$ for $\pi$ and
Lorentzian metrics $h$ on $E$ such that $\pi:(E,h)\rightarrow (M,g)$
is a Kaluza-Klein space and such that $h_{m_0}=\cG$. Let:
\be
G\eqdef \{T_\gamma\, , \,\, |\,\, \gamma\in \cP(M)\, , \quad \gamma(0)=\gamma(1)=m_0\}\subseteq \Iso(\cM,\cG)
\ee
be the Ehresmann holonomy group at $m_0$. For any $m\in M$, consider the set:
\be
\Pi_m \eqdef \{T_\gamma\,\, |\,\, \gamma\in \cP(M): \gamma(0)=m_0\,\, \&\,\,  \gamma(1)=m\}\, ,
\ee
and let $\Pi\eqdef \sqcup_{m\in M}\Pi_m$ be endowed with the obvious
projection to $M$.  Then $\Pi$ is a principal $G$-bundle (known as the
{\em holonomy bundle} of $H$ relative to $m_0$) under the obvious
right action of $G$. Moreover, the Ehresmann connection $H$ induces a
principal connection $\theta$ on $\Pi$. Conversely, let $\rho$ be the
isometric action of $G$ on $\cM$ given by the inclusion $G\subseteq
\Iso(\cM,\cG)$. Then the associated bundle construction associates to
a principal bundle $\Pi$ with principal connection $\theta$ the fiber
bundle $E_\Pi\eqdef \Pi\times_\rho \cM$ with induced Ehresmann
connection $H_\theta$. The two correspondences described above give
mutually quasi-inverse functors between the groupoid of Kaluza-Klein
spaces over $(M,g)$ which have typical fiber $(\cM,\cG)$ and are
endowed with a horizontal distribution $H$ with holonomy $G\subseteq
\Iso(\cM,\cG)$ and the groupoid of principal $G$-bundles $\Pi$ defined
over $M$ and endowed with a principal connection $\theta$ together
with an embedding $G\subseteq \Iso(\cM,\cG)$ up to conjugation. Hence:

\begin{proposition}
Let $g$ be a fixed Lorentzian metric on $M$. Then isomorphism classes
of Kaluza-Klein spaces over $(M,g)$ with typical fiber $(\cM,\cG)$ and
horizontal distribution $H$ having Ehresmann holonomy contained in
$\Iso(\cM,\cG)$ are in bijection with isomorphism classes of principal
$G$-bundles $\Pi$ defined over $M$ and endowed with a principal
connection $\theta$ whose holonomy is embedded in $\Iso(\cM,\cG)$. 
Moreover, $H$ is integrable if and only if $\theta$ is
flat.
\end{proposition}

\noindent
Let $(\cM,\cG)$ be a Riemannian manifold. Let $E^{0} \eqdef
M\times\cM$ and let $\pi^0:M\times \cM\rightarrow M$ and
$p^0:M\times\cM\rightarrow \cM$ denote the canonical
projections. Notice that $\pi^0$ is the trivial fiber bundle over $M$
with fiber $\cM$. Let $V^0=(T\cM)^{p^0}$ denote the vertical
distribution of $\pi^0$, which is endowed with the pull-back metric
$\cG^{p^0}$.

\begin{definition}
A {\em topologically trivial Kaluza-Klein space} over $(M,g)$ with
fiber $(\cM,\cG)$ is a Kaluza-Klein space of the form
$\pi^0:(E^0,h^{0})\rightarrow (M,g)$, whose vertical metric is the
pull-back metric $h_{V^0}\eqdef \cG^{p^0}$.
\end{definition}

\noindent 
The metric $h^{0}=h(g,\cG,H)$ of such a space is uniquely determined by
$g$, $\cG$ and by the horizontal distribution $H$, whose parallel
transport $T$ must preserve $\cG^{p^0}$. For any $\gamma\in \cP(M)$
and all $p\in \cM$, we define:
\ben
\label{hatT}
T_\gamma(\gamma(0),p)=(\gamma(1),{\hat T}_\gamma(p))~~, 
\een
where ${\hat T}_\gamma\in \Iso(\cM,\cG)$ satisfy ${\hat T}_{\gamma_1\gamma_2}={\hat T}_{\gamma_1}\circ {\hat T}_{\gamma_2}$.

\begin{definition}
The {\em product Kaluza-Klein space} over $(M,g)$ with fiber
$(\cM,\cG)$ is the topologically trivial Kaluza-Klein space
$\pi^0:(E^0,h^0_{\triv})\rightarrow (M,g)$ with fiber $(\cM,\cG)$
whose horizontal distribution is the trivial integrable Ehresmann
connection $H^0_{\triv} \eqdef (T M)^{\pi^0}$, with induced horizontal
metric $h^0_{\triv}$.
\end{definition}

\noindent 
For any $p\in \cM$, the horizontal lift through $(\gamma(0),p)\in
E^0_{\gamma(0)}$ of any path $\gamma\in \cP(M)$ defined by the
trivial Ehresmann connection $H^0$ is the path given by
$\bgamma_{p, \triv}(s)=(\gamma(s),p)$ for all $s\in [0,1]$. Hence the
Ehresmann transport $T^0_\gamma$ defined by $H^0$ is given by ${\hat
T}^0_{\gamma}=\id_{\cM}$.  This implies that the Kaluza-Klein metric
$h^0_{\triv}=h_{g, \cG, H^0_\triv}$ equals the product metric
$g\times \cG$.

As explained above, integrable Kaluza-Klein spaces correspond to flat
principal $G$-bundles $(\Pi,\theta)$ defined over $M$ together with an
embedding $G\subseteq \Iso(\cM,\cG)$. It is well-known that flat
principal $G$-bundles are classified up to isomorphism by their
holonomy morphism $\Hol_{P,\theta}:\pi_1(M)\rightarrow G$, considered
up to the conjugation action of $G$. Hence the set of isomorphism
classes of integrable Kaluza-Klein spaces over $(M,g)$ with
fiber $(\cM,\cG)$ and Ehresmann holonomy contained in 
$\Iso(\cM,\cG)$ is in bijection with the character variety:
\begin{equation*}
\cM_{\Iso(\cM,\cG)}(M)\eqdef \Hom(\pi_1(M),\Iso(\cM,\cG))/\Iso(\cM,\cG)\, .
\end{equation*}
When the space-time $M$ is simply-connected, any integrable
Kaluza-Klein space over $(M,g)$ is isomorphic with a product space
(see below). In that case, we necessarily have $G=1$ (the trivial
group) and $H$ is gauge-equivalent with the trivial Ehresmann
connection.

\subsection{Vertical scalar potentials and scalar bundles}

Let $\pi:(E,h)\rightarrow (M,g)$ be a Kaluza-Klein space and $T$ be
its Ehresmann transport.

\begin{definition}
A \emph{vertical scalar potential} for $\pi$ is a smooth $T$-invariant
real-valued function $\bPhi \in \cC^\infty(E,\R)$ defined on the total
space $E$ of $\pi$.
\end{definition}

\noindent 
The $T$-invariance of $\bPhi$ means that the restrictions
$\Phi_m\eqdef \bPhi|_{E_m}\in \cC^\infty(E_m,\R)$ to the fibers of $E$
satisfy:
\be
\Phi_{\gamma(1)}\circ T_\gamma=\Phi_{\gamma(0)}\, ,\quad \forall\,\, \gamma\in \cP(M)~~, 
\ee
a condition which is equivalent with the requirement that $\bPhi$ is
annihilated by any horizontal vector field defined on $E$:
\ben
\label{Vvert}
X(\bPhi)=0\, , \quad \forall\, X\in \Gamma(E,H)\, .
\een
This implies that all fiber restrictions $\Phi_m$ ($m\in M$) can be
recovered from $\Phi_{m_0}$, where $m_0$ is any fixed point of $M$. In
particular, the isomorphism type of the scalar structure\footnote{As
defined in reference \cite{GESM}.} $(E_m,h_m,\Phi_m)$ is independent
of $m$. The holonomy group $G_{m}\subset \Iso(E_{m},h_{m})$ of $\pi$
at any point $m\in M$ preserves $\Phi_{m}$:
\be
G_m \subset \Iso(E_m, h_m, \Phi_m)\eqdef \{\varphi\in \Iso(E_m,h_m)|\Phi_m\circ \varphi=\Phi_m\}~~.
\ee
Relation \eqref{Vvert} amounts to $\dd \bPhi (X)=0$ for all
$X\in\Gamma(X,H)$, i.e. $\dd \bPhi \circ P_H=0$. Since
$P_V+P_H=\id_{TE}$, this gives $\dd \bPhi=\dd \bPhi \circ P_V$, which
shows that $\dd \bPhi$ can be viewed as an element of
$\Gamma(E,V^\ast)$. Since $H$ and $V$ are $h$-orthogonal, this implies
that the gradient of $\bPhi$ is a vertical vector field:
\be
\grad_h \bPhi\in \Gamma(E,V)~~.
\ee

\begin{definition}
A {\em scalar bundle} over $(M,g)$ is a pair $(\pi:(E,h)\rightarrow
(M,g),\bPhi)$, where $\pi$ is a Kaluza-Klein space and $\bPhi$ is a
vertical scalar potential for $\pi$. The scalar bundle is called {\em
integrable} if the Kaluza-Klein space $\pi$ is integrable.
\end{definition}

\noindent 
The isomorphism type of the scalar structures $(E_m,h_m,\Phi_m)$
(which, as explained above, does not depend on the point $m\in M$) is
called the {\em type} of the scalar bundle and will be generally
denoted by $(\cM,\cG,\Phi)$.  The classification of integrable
Kaluza-Klein spaces immediately implies the following.

\begin{proposition}
Integrable scalar bundles defined over $(M,g)$ and having type
$(\cM,\cG,\Phi)$ are classified up to isomorphism by the points of the
character variety:
\be
\cM_{\Iso(\cM,\cG,\Phi)}(M)\eqdef \Hom(\pi_1(M),\Iso(\cM,\cG,\Phi))/\Iso(\cM,\cG,\Phi)~~.
\ee
\end{proposition}

\subsection{Special trivializing atlases for a Kaluza-Klein space and for a scalar bundle}
\label{spectriv}
Recall that an open subset $U$ of a pseudo-Riemannian
manifold is called {\em geodesically convex} \cite{ONeillBook} if it
is a normal neighborhood for each of its points. If $U$ is convex,
then for any two points $p,q\in U$ there exists a unique geodesic
segment which is contained in $U$ and which connects $p$ and $q$.  Any
point of a pseudo-Riemannian manifold has a basis of geodesically
convex neighborhoods (see \cite[p. 129]{ONeillBook}). A {\em convex
  cover} of a pseudo-Riemannian manifold is an cover by open and
geodesically convex sets which has the property that any nontrivial
intersection of two of its elements is geodesically convex. Given
any open cover $\mathfrak{V}$ of a pseudo-Riemannian manifold, there
exists a convex cover $\mathfrak{U}$ such that any element of
$\mathfrak{U}$ is contained in some element of $\mathfrak{V}$
\cite[Lemma 5.10]{ONeillBook}. In particular, any pseudo-Riemannian 
manifold admits convex covers. 

Let $\pi:(E,h)\rightarrow (M,g)$ be a Kaluza-Klein space. Let $m_0 \in
M$ be a fixed point and set $\cM\eqdef E_{m_0}$ and $\cG\eqdef
h_{m_0}$. Let $\mathfrak{U}=\left( U_\alpha\right)_{\alpha \in I}$ be
a convex cover for $(M,g)$, where the indexing set is chosen such that
$0\not \in I$. Fix points $m_{\alpha}\in U_\alpha$ and paths
$\lambda^{\alpha}\in \cP(M)$ such that
$\lambda^{\alpha}(0)=m_{\alpha}$ and $\lambda^{\alpha}(1)=m_0$. For
any $m\in U_\alpha$, let $\gamma_m^{\alpha}:[0,1]\rightarrow U_\alpha$
be the unique smooth geodesic contained in $U_\alpha$ such that
$\gamma_m^{\alpha}(0)=m$ and $\gamma_m^{\alpha}(1)=m_\alpha$. For any
$\alpha\in I$, let $g_\alpha \eqdef g|_{U_\alpha}$, $E_{\alpha} \eqdef
E|_{U_\alpha}$, $\pi_\alpha \eqdef \pi|_{U_\alpha}$. Let
$E^0_{\alpha}\eqdef U_\alpha\times \cM$ and $\pi^0_{\alpha}\colon
E_{\alpha}^{0}\rightarrow U_\alpha$, $p^0_{\alpha}\colon
E^0_{\alpha}\rightarrow \cM$ be the projections on the first and
second factors. Then $\pi^0_{\alpha}\colon E^0_{\alpha}\eqdef
U_\alpha\times \cM\rightarrow U_\alpha$ is the trivial fiber bundle
over $U_\alpha$ with fiber $\cM$. Define diffeomorphisms
$q_\alpha:E_{\alpha}\rightarrow U_\alpha\times \cM$ through:
\ben
\label{qalpha}
q_\alpha(e)=(\pi(e), {\hat q}_\alpha(e))\, , \quad  \forall\, e\in E_\alpha\, ,
\een
where ${\hat q}_\alpha\colon E_\alpha\rightarrow \cM$ is given by the
following differentiable surjective map:
\ben
\label{hatqalpha}
\hat{q}_\alpha(e)\eqdef p^0_{\alpha}\circ (T_{\lambda^{\alpha}}\circ T_{\gamma^{\alpha}_{\pi(e)}})(e)=p^0_{\alpha}\circ T_{\lambda^{\alpha}\circ \gamma^{\alpha}_{\pi(e)}}(e)\in  \cM\, , \qquad \forall\, e\in E_\alpha\, .
\een
Here we have used the identification $\cM\eqdef E_{m_{0}}$ and the fact that:
\begin{equation*}
(T_{\lambda^{\alpha}}\circ T_{\gamma^{\alpha}_{\pi(e)}})(e) \in \left\{m_{0}\right\}\times E_{m_{0}}\, , \qquad \forall\, e\in E_{\alpha}\, .
\end{equation*}
The restriction ${\hat q}_{\alpha}(m)\eqdef {\hat
  q}_\alpha|_{E_m}:E_m\rightarrow \cM$ to the fiber at $m\in U_\alpha$
is an isometry from $(E_m,h_m)$ to $(\cM,\cG)$ which is given
explicitly by:
\begin{equation*}
{\hat q}_{\alpha}(m)(e_{m})=p^{0}_{\alpha}\circ T_{\lambda^{\alpha}\circ \gamma^{\alpha}_{m}}(m\times e_{m})\, , \quad \forall\, e_{m}\in E_m\, . 
\end{equation*}
The maps $q_\alpha$ are diffeomorphisms from $E_\alpha$ to
$E^0_{\alpha}=U_\alpha\times\cM$ which fit into the following
commutative diagram:
\ben
\label{ltriv}
\scalebox{1.2}{
	\xymatrix{
		E_\alpha ~~\ar[d]_{\pi_\alpha} \ar[r]^{\!\!\!\!\!\!q_\alpha} & ~~E^0_{\alpha} \ar[d]^{\pi^0_{\alpha}}\\
		U_\alpha  \ar[r]^{\id_{U_\alpha}} & U_\alpha \\
	}}
\een
Hence $(U_\alpha, q_\alpha)$ is a trivializing atlas for the fiber
bundle $\pi$, called the {\em special trivializing atlas} determined
by the convex cover $(U_\alpha)_{\alpha\in I}$, by the reference
point $m_0$ and by the choices of points $m_{\alpha}\in U_\alpha$ and
of paths $\lambda^{\alpha}$ from $m_{\alpha}$ to $m_0$.
	
Let $h_{\alpha}\eqdef h|_{E_{\alpha}}$. Since the Ehresmann transport
$T$ is isometric, the vertical metric $h_\alpha|_{V}$ agrees with
$\cG^{p^0_{\alpha}}$ through the diffeomorphism $q_\alpha\colon
E_{\alpha}\rightarrow U_\alpha\times \cM$. Hence $h_\alpha$
corresponds through $q_\alpha$ to a Kaluza-Klein metric
$h^{0}_{\alpha}$ on the trivial bundle
$\pi^0_{\alpha}:E_{\alpha}\rightarrow U_\alpha$. The latter is the
Kaluza-Klein metric determined by $g|_{U_\alpha}$, $\cG$ and the
distribution $H^{0}_{\alpha} \eqdef (\dd q_\alpha)(H_\alpha)$, where
$H_\alpha\eqdef H|_{E_\alpha}$. The diffeomorphism $q_\alpha$ is an
isometry from $(E_\alpha,h_\alpha)$ to $(E^0_{\alpha},
h^{0}_{\alpha})$ which makes diagram \eqref{ltriv} into an isomorphism
of Kaluza-Klein spaces from the Kaluza-Klein space
$\pi_\alpha:(E_\alpha,h_\alpha)\rightarrow (U_\alpha,g_\alpha)$ to the
Kaluza-Klein space $\pi^0_{\alpha}\colon
(E^0_{\alpha},h^{0}_{\alpha})\rightarrow (U_\alpha,g_\alpha)$. Notice
that the second of these need not be a product Kaluza-Klein space,
since $h^{0}_{\alpha}$ may differ from the product metric
$g_\alpha\times \cG$.
	
Since $(M,g)$ admits convex covers, any Kaluza-Klein space admits
special trivializing atlases. In particular, any Kaluza-Klein space is
locally isomorphic with a topologically trivial Kaluza-Klein space
(which need not be a product Kaluza-Klein space!).
	
For any $\alpha,\beta\in I$ such that $U_{\alpha\beta}\eqdef
U_\alpha\cap U_\beta\neq \emptyset$, we have:
\be
(q_\beta\circ q_\alpha^{-1})(m,p)=(m, \bg_{\alpha\beta}(m)(p))\, , \quad \forall\, m\in U_{\alpha\beta}\, , \quad \forall\, p\in \cM\, ,
\ee
where the transition functions $\bg_{\alpha\beta}:U_{\alpha\beta}\rightarrow \Iso(\cM,\cG)$ are given by:
\ben
\label{bgdef}
\bg_{\alpha\beta}(m)\eqdef {\hat q}_\beta(m)\circ {\hat q}^{-1}_\alpha(m)\in \Iso(\cM,\cG)\, , \quad \forall\, m\in U_{\alpha\beta}\, .
\een
Using \eqref{hatqalpha}, this gives:
\be
\bg_{\alpha\beta}(m)=T_{\lambda^{\beta}}\circ T_{\gamma^{\beta}_{m}}\circ T_{\gamma^{\alpha}_{m}}^{-1}\circ T_{\lambda^{\alpha}}^{-1}=
T_{c^{\alpha\beta}_m}\in \Iso(\cM,\cG)\, , \quad \forall\, m\in U_{\alpha\beta}\, .
\ee
Here: 
\ben
\label{calphabeta}
c^{\alpha\beta}_m\eqdef
\lambda^{\beta}\circ\gamma^{\beta}_{m}\circ (\gamma^{\alpha}_{m})^{-1}\circ (\lambda^{\alpha})^{-1}\, ,
\een 
is the closed path starting and ending at $m_0$ and passing through the
point $m\in U_{\alpha\beta}$ which is shown in
Figure \ref{fig:specialtriv}.
\begin{figure}[h]
\centering
\scalebox{0.4}{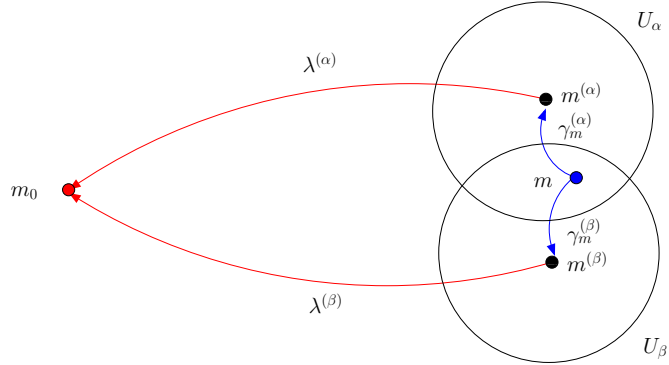}
\caption{The transition functions of a special trivializing atlas are
  determined by closed paths based at the reference point $m_0\in M$
  and passing through $m\in U_{\alpha\beta}$.}
\label{fig:specialtriv}
\end{figure}

Let $\bPhi$ be a vertical potential for $\pi$. In this case, a
special trivializing atlas for the Kaluza-Klein space $\pi$ is also
called a special trivializing atlas for the scalar bundle
$(\pi,\bPhi)$. Let $\bPhi_{\alpha} \eqdef \bPhi|_{E_\alpha}$ and set
$\Phi\eqdef \Phi_{m_0}$. Since $\bPhi$ is $T$-invariant, the
definition \eqref{hatqalpha} of ${\hat q}_\alpha$ implies:
\ben
\label{bPhialpha}
\bPhi_\alpha=\Phi\circ {\hat q}_\alpha\, .
\een
This gives $\Phi\circ \bg_{\alpha\beta}=\Phi$,
i.e. $\bg_{\alpha\beta}\in \Iso(\cM,\cG,\Phi)$.

When the Kaluza-Klein space $\pi$ is integrable, the Ehresmann
transport depends only on the homotopy class of curves in $M$.  In
this case, $T_{c^{\alpha\beta}_m}$ is independent of the point $m\in
U_\alpha\cap U_\beta$ since $U_\alpha\cap U_\beta$ is path connected
and hence the homotopy class of $c^{\alpha\beta}_m$ does not depend on
$m$. Moreover, integrability of $H$ implies that $H_{\alpha}^{0}$
coincide with the trivial Ehresmann connections $H_{\alpha}^{\triv}$
of $\pi^0_{\alpha}$.  In this case, $\pi^0_{\alpha}\colon
(E^0_{\alpha}, h_\alpha)\rightarrow (U_\alpha,g_\alpha)$ is a product
Kaluza-Klein space and $q_\alpha$ is an isometry from $(E_{\alpha},
h_\alpha)$ to $(U_\alpha\times\cM,g_\alpha\times \cG)$. Thus:
	
\begin{proposition}
Let $\pi:(E,h)\rightarrow (M,g)$ be an integrable Kaluza-Klein
space. Then the local trivializing maps $q_\alpha$ of a special
trivializing atlas consist on isometries from $(E_{\alpha}, h_\alpha)$
to $(U_\alpha\times \cM,g_\alpha\times\cG)$ and the transition
functions $\bg_{\alpha\beta}$ defined by such an atlas are constant on
$U_\alpha\cap U_\beta$. In particular, $\pi$ is locally isomorphic
with a product Kaluza-Klein space.
\end{proposition}

\noindent
In this case, relation \eqref{bgdef} gives: 
\ben
\label{bgflat}
{\hat q}_\beta|_{E_{\alpha\beta}}=\bg_{\alpha\beta}\circ {\hat q}_\alpha|_{E_{\alpha\beta}}~~\forall \alpha,\beta\in I~~,
\een
where $E_{\alpha\beta}\eqdef E|_{\pi^{-1}(U_{\alpha\beta})}$ and we
define $\bg_{\alpha\beta}\eqdef \id_\cM$ when
$U_{\alpha\beta}=\emptyset$.
	
\begin{remark}
Suppose that the integrable Kaluza-Klein space $\pi$ is endowed with a
vertical potential $\bPhi$ whose restriction to $\cM$ we denote by
$\Phi$. Then the constant transition functions in a special
trivializing atlas satisfy $\bg_{\alpha\beta}\in \Iso(\cM,\cG,\Phi)$.
\end{remark}

\section{Section sigma models}
\label{sec:scalarsection}

Given a Kaluza-Klein space $\pi:(E,h)\rightarrow (M,g)$, let
$P_V:TE\rightarrow V$ and $P_H:TE\rightarrow H$ denote the
corresponding $h$-orthogonal projectors and $T$ denote the Ehresmann
transport of $H\eqdef H(h)$. Let $\nabla^v\eqdef P_V\circ \nabla$
denote the connection induced on $V$ by the Levi-Civita connection
$\nabla$ of $(E,h)$. Let $\bPhi\in \cC^\infty(E,\R)$ be a vertical
scalar potential for $\pi$, so that $(\pi,\bPhi)$ is a bundle of
scalar structures.

\subsection{The scalar section sigma model defined by $\pi$ and $\bPhi$}

\begin{definition}
The {\em vertical Lagrange density} of $\pi$ is the map
$e^v_\Phi:\Gamma(\pi)\rightarrow \cC^\infty(M,\R)$ defined, for every
$s\in \Gamma(\pi)$, as follows:
\be
e^v_\bPhi(g,h,s)\eqdef \frac{1}{2}\Tr_{g} s^\ast_v(h_V)+\bPhi^s~~,
\ee 
where $s^\ast_v(h_V)$ is the vertical first fundamental form of $s$
(see Appendix \ref{app:submersions}) and $\bPhi^s= \bPhi\circ s\in
\cC^\infty(M,\R)$. Let:
\be
e^v(g,h,s)\eqdef e^v_0(g,h,s)=\frac{1}{2}\Tr_{g} s^\ast(h_V)~~.
\ee
\end{definition}

\begin{definition}
The {\em section sigma model} defined by $\pi$ and $\bPhi$ has action
functional $S_{\sc}\colon\Met_{3+n,1}(E)\times \Gamma(\pi)\rightarrow
\R$ given by:
\ben
\label{Ssc}
S_{\sc,U}[g, h, s]=-\int_U\nu_M(g)e^v_\bPhi(g,h,s)
\een
for any relatively-compact subset $U\subset M$. 
\end{definition}

\noindent
Let $s\in\Gamma(\pi)$ be a section. The differential $\dd s\colon
TM\to TE$ of $s$ is an unbased morphism of vector bundles equivalent
to a section $\dd s\in \Omega^{1}(M, TE^{s})$ which for simplicity we
denote by the same symbol. We define the {\em vertical differential}
$\dd^{v}s\eqdef P^{s}_{V}\circ \dd s\in\Omega^{1}(M, V^{s})$, where
$P^{s}_{V}$ denotes the vertical projection of the pull-back bundle
(see Appendix \ref{app:submersions}). The Levi-Civita connection on
$(M,g)$ and the $s$-pull-back of the connection $\nabla^{v}$ on $V$
induce a connection on $T^{\ast}M\otimes V^{s}$, which for simplicity
we denote again by $\nabla^{v}$.

\begin{definition}
The {\em vertical tension field} of $s\in \Gamma(\pi)$ is defined
through:
\be
\label{tensionvert}
\tau^v(g, h, s)\eqdef \Tr_{g}\nabla^v \dd^v s \in \Gamma(M,V^s)~~.
\ee
\end{definition}

\begin{remark}
The ordinary Lagrange density and tension field of $s\in\Gamma(\pi)$
are defined similarly to their vertical counterparts except that there
is no vertical projection involved. More precisely, we define:
\be
e_\bPhi(g,h,s)\eqdef \frac{1}{2}\Tr_{g} s^\ast (h)+\bPhi^s\, ,
\ee
as well as:
\be
\tau(g, h, s)\eqdef \Tr_{g}\nabla \dd s \in \Gamma(M,TE^s)\, ,
\ee
where $\nabla$ denotes the connection on $T^{\ast}M\otimes TE^{s}$
induced by the Levi-Civita connection on $(M,g)$ and the $s$-pullback
of the Levi-Civita connection on $(E,h)$. For simplicity we will
sometimes drop the explicit dependence on $g$ and $h$ in
$e_{\bPhi}(g,h,s)$, $\tau(g, h, s)$ etc.
\end{remark}

\noindent
The following is an easy adaptation of a result due to \cite{Wood1}.

\begin{proposition}
The vertical density $e^v(s)$ differs from $e(s)$ by a constant and we
have $\tau^v(s)=P_V\circ\tau(s)$. Moreover, the critical points
of \eqref{Ssc} with respect to $s\in \Gamma(\pi)$ are solutions of the
equation:
\ben
\label{sceom}
\tau^v(s)=-(\grad_h\bPhi)^s\, .
\een
\end{proposition}

\begin{remark}
When $\bPhi=0$, equation \eqref{sceom} becomes the {\em vertical
pseudoharmonic equation}:
\be
\label{pseudoharm}
\tau^v(s)=0
\ee
and its solutions are called {\em pseudoharmonic sections} of $\pi$. 
\end{remark}	
	
\noindent 
Notice that \eqref{Ssc} is extremized only with respect to vertical
variations of $s$, since $s$ is subject to the section constraint
$\pi\circ s=\id_M$. As a consequence, a section which is
pseudo-harmonic as an unconstrained map from $(M,g)$ to $(E,h)$ is a
pseudo-harmonic section, but not every pseudo-harmonic section is
pseudo-harmonic as an unconstrained map from $(M,g)$ to $(E,h)$.

\subsection{The sheaves of configurations and solutions}

The local character of the model allows us to define two sheaves of
sets on $M$, namely:

\begin{itemize}

\item The sheaf of configurations $\Conf_\pi$, which coincides with
  the sheaf of local smooth sections of $\pi$.

\item The sheaf of solutions $\Sol_{\pi,\Phi}$, defined as the
  sub-sheaf of $\Conf_\pi$ whose set of sections $\Sol_{\pi,\Phi}(U)$
  over an open subset $U\subset M$ consist of pseudo-harmonic sections
  of the restricted Kaluza-Klein space $\pi_U:(E_U,h_U)\rightarrow
  (U,g_U)$ endowed with the vertical potential $\Phi|_{E_U}$, where
  $E_U\eqdef \pi^{-1}(U)$, $h_U\eqdef h|_{E_U}$, $\pi_U\eqdef
  \pi|_{E_U}$ and $g_U\eqdef g|_{U}$.
\end{itemize}

\subsection{A modified sigma model for maps}

Let $\pi^0\colon (E^0\eqdef M\times \cM,h^0)\rightarrow (M,g)$ be a
topologically-trivial Kaluza-Klein space over $(M,g)$ endowed with a
vertical scalar potential $\bPhi\in \cC^\infty(E,\R)$. Let $H^0$ be
the horizontal distribution determined by $h^0$. Let
$\graph\colon \cC^\infty(M,\cM)\rightarrow \Gamma(\pi^0)$ be the
bijective map given by:
\be
\graph(\varphi)(m)\eqdef (m,\varphi(m))~~\forall m\in M~~,
\ee
whose inverse is the map $\ungraph:\Gamma(\pi^0)\rightarrow \Gamma(\pi^0)$ given by:
\be
\ungraph(s^0)\eqdef p^0\circ s^0\, .
\ee
Using this correspondence, the scalar section sigma model defined by
$\pi^0\colon (E^0,h^0)\rightarrow (M,g)$ together with the vertical
potential $\bPhi\in \cC^\infty(E^0,\R)$ can be viewed as a
generalization of the ordinary scalar sigma model of maps from $(M,g)$
to $(\cM,\cG,\Phi)$.
	
\begin{definition}
The {\em modified scalar sigma model} determined by $(M,g)$,
$(\cM,\cG)$ and $H^0$ is defined by the action:
\be
S^{H^0}_{\sc,\bPhi, U^0}[\varphi]=S_{\sc,\pi^0, \bPhi, U^0}[\graph(\varphi)]\, .
\ee
for any relatively compact open set $U^0\subset M$, where
$S_{\sc, \pi, \bPhi, U^0}$ is the action of the section sigma model of
the topologically trivial Kaluza-Klein space $\pi^0\colon
(E^0,h^0)\rightarrow (M,g)$ with fiber $(\cM,\cG)$ and horizontal
distribution $H^0$.
\end{definition}

\begin{definition}
Let $H^0$ be any horizontal distribution for the trivial bundle
$\pi^0\colon E^0\rightarrow M$.  A map $\varphi:(M,g)\rightarrow
(\cM,\cG)$ is called {\em $H^0$-pseudoharmonic} if the graph $s^0$ of
$\varphi$ is a pseudo-harmonic section of the topologically trivial
Kaluza-Klein space $\pi^0\colon (E^0,h^{0})\rightarrow (M,g)$ with
fiber $(\cM,\cG)$, where $h^{0}$ is the metric on $E^{0}$ determined
by $H^{0}, \cG$ and $g$.
\end{definition}

\noindent 
It is clear that the solutions of the equations of motion of the
modified scalar sigma model are $H^0$-pseudoharmonic maps. The
modified scalar sigma model defined reduces to the ordinary sigma
model when $H^0$ is the trivial Ehresmann connection of $\pi^0$, as we
explain next.

\begin{remark}
Let $\pi^0:(E^0\eqdef M\times\cM, h^0\eqdef g\times \cG)\rightarrow
(M,g)$ be the product Kaluza-Klein space with fiber $(\cM,\cG)$
defined over $(M,g)$ and horizontal distribution $H^{0}_{\triv}$. For
any $s\in \Gamma(\pi^0)$, we have $s^\ast(h)=\varphi^\ast (\cG)$ and
$\dd^v s=\dd \varphi$, where $\varphi=\ungraph(s)=p^0\circ s$. Thus
$e^v(s)=e(\varphi)$ and $\tau^v(s)=\tau(\varphi)$. Moreover, we have
$\bPhi=\Phi\circ p^0$ for some $\Phi\in \cC^\infty(\cM,\R)$. Hence the
section sigma model action \eqref{Ssc} reduces to the action of an
ordinary sigma model of maps from $(M,g)$ to $(\cM,\cG)$, while the
equations of motion \eqref{sceom} reduce to those of an ordinary sigma
model. Setting $\Phi=0$, we conclude that a map
$\varphi:(M,g)\rightarrow (\cM,\cG)$ is $H^0_\triv$-pseudoharmonic if
and only if it is pseudo-harmonic.
\end{remark}

\subsection{U-fold interpretation of globally-defined solutions in the integrable case}
\label{subsec:sigmaUfold}

Let $(\pi:(E,h)\rightarrow (M,g), \bPhi)$ be an {\em integrable}
bundle of scalar data of type $(\cM,\cG,\Phi)$. Consider a special
trivializing atlas of $\pi$ defined by the convex cover
$(U_\alpha)_{\alpha\in I}$ of $(M,g)$.  We freely use the notations
introduced in Subsection \ref{spectriv}. Since $\pi$ is integrable,
the trivializing maps defined in \eqref{qalpha} give isometries
$q_\alpha:(E_\alpha,
h_\alpha)\stackrel{\sim}{\rightarrow}(U_\alpha\times \cM,
g_\alpha\times \cG)$. For any pair of indices $\alpha,\beta\in I$ such
that $U_{\alpha\beta}\eqdef U_\alpha\cap U_\beta$ is non-empty, the
composition $q_{\alpha\beta}\eqdef q_\beta\circ
q_\alpha^{-1}:U_{\alpha\beta}\times \cM\rightarrow
U_{\alpha\beta}\times \cM$ has the form
$q_{\alpha\beta}(m,p)=(m,\bg_{\alpha\beta}(p))$, where:
\be
\bg_{\alpha\beta}\in \Iso(\cM,\cG,\Phi)\, .
\ee
We remind the reader that
$\Phi\eqdef \Phi_{m_0}=\bPhi|_{E_{m_0}}$. Setting
$\bg_{\alpha\beta}=\id_\cM$ for $U_{\alpha\beta}=\emptyset$, the
collection $(\bg_{\alpha\beta})_{\alpha,\beta\in I}$ satisfies the
cocycle condition:
\ben
\label{gcocycle}
\bg_{\beta\delta}\bg_{\alpha\beta}=\bg_{\alpha\delta}\, , \quad \forall\, \alpha,\beta,\delta\in I\, .
\een

\noindent
For any section $s\in \Gamma(\pi)$, the restriction $s_\alpha\eqdef
s|_{U_\alpha}$ corresponds through $q_\alpha$ to the graph
$\graph(\varphi^{\alpha})\in \Gamma(\pi^0_{\alpha})$ of a
uniquely-defined smooth map
$\varphi^{\alpha}\in\cC^\infty(U_\alpha,\cM)$:
\ben
\label{sres}
s_\alpha=q_{\alpha}^{-1}\circ \graph(\varphi^\alpha)~~\mathrm{i.e.}~~q_\alpha(s_\alpha(m))=(m,\varphi^\alpha(m))\, , ~~\forall m\in U_\alpha\, .
\een 
Composing the first relation from the left with $p^0_{\alpha}$ gives:
\be
\varphi^{\alpha}={\hat q}_\alpha\circ s_\alpha~~.
\ee
Using relation \eqref{bgflat}, this implies:
\ben
\label{phigluing}
\varphi^{\beta}(m)=\bg_{\alpha\beta}\varphi^{\alpha}(m)~~\forall m\in U_{\alpha\beta}~~,
\een
where juxtaposition in the right hand side denotes the tautological
action of the group $\Iso(\cM,\cG,\Phi)$ on $\cM$. Conversely, any
family of smooth maps
$\left\{\varphi^{\alpha}\in\cC^\infty(U_\alpha,\cM)\right\}_{\alpha\in
I}$ satisfying \eqref{phigluing} defines a smooth section
$s\in \Gamma(\pi)$ whose restrictions to $U_\alpha$ are given
by \eqref{sres}.
	
From the previous discussion, the equation of motion \eqref{sceom} for
$s$ is equivalent with the condition that each $\varphi^{\alpha}$
satisfies the equation of motion of the ordinary sigma model defined
by the scalar data $(\cM,\cG,\Phi)$ on the space-time
$(U_\alpha,g_\alpha)$:
\ben
\label{localeomscalar}
\tau^v(h, s)=-(\grad \bPhi)^s \Leftrightarrow \tau(g_\alpha, \varphi^{\alpha})=-(\grad\Phi)^{\varphi^{\alpha}} ~~\forall \alpha\in I~~.
\een
Thus global solutions $s$ of the equations of motion \eqref{sceom} are
\emph{glued} from local solutions $\varphi^{\alpha}\in
\cC^\infty(U_\alpha,\cM)$ of the equations of motion of the ordinary
sigma model using the $\Iso(\cM,\cG,\Phi)$-valued constant transition
functions $\bg_{\alpha\beta}$ which satisfy the cocycle condition
\eqref{gcocycle}. This realizes the ideology of {\em
  classical}\footnote{As opposed to string-theoretical.} U-folds,
namely \emph{gluing local solutions through symmetries of the
  equations of motion}.

\subsection{Sheaf-theoretical description} 
	
The observations above have the following sheaf-theoretical
description. Let $\Conf_\cM^{\alpha}\eqdef \Conf_\cM|_{U_\alpha}$,
$\Sol^{g,{\alpha}}_{\cM,\cG,\Phi}\eqdef
\Sol^g_{\cM,\cG,\Phi}|_{U_\alpha}$ and
$\fg_{\alpha\beta}:\Conf_\cM^{\alpha}|_{U_{\alpha\beta}}\rightarrow\Conf_\cM^{\beta}|_{U_{\alpha\beta}}$
be the isomorphism of sheaves defined through:
\be
\fg_{\alpha\beta}(s) \eqdef \bg_{\alpha\beta} s\, ,\quad  \,\, \forall\, U\subset U_{\alpha\beta}\, , \quad \forall s\in \Conf_\cM^{\alpha}(U)~~.
\ee
Since $\bg_{\alpha\beta}$ acts by symmetries of the equations of
motion of the ordinary sigma model, this restricts to an isomorphism
of sheaves of sets from
$\Sol^{g,{\alpha}}_{\cM,\cG,\Phi}|_{U_{\alpha\beta}}$ to
$\Sol^{g,\beta}_{\cM,\cG,\Phi}|_{U_{\alpha\beta}}$. These isomorphisms
of sheaves satisfy the cocycle conditions:
\be
\fg_{\beta\gamma}\fg_{\alpha\beta}=\fg_{\alpha\gamma}\, , \quad\forall \alpha,\beta,\gamma\in I\, . 
\ee
The sheaves $\Conf_\cM^{\alpha}$ and $\Sol_{\cM,\cG,\Phi}^{g,\alpha}$
defined on $U_\alpha$ glue using these isomorphisms to sheaves
$\bConf_\cM$ and $\bSol^{g}_{\cM,\cG,\Phi}\subset \bConf_\cM$ defined
on $M$ which satisfy $\bConf_\cM|_{U_\alpha}\simeq \Conf_\cM^{\alpha}$
and
$\bSol_{\cM,\cG,\Phi}^{g}|_{U_\alpha}\simeq \Sol_{\cM,\cG,\Phi}^{g,\alpha}$. The
discussion above shows that the trivialization maps $q_\alpha$ of
$\pi$ induce isomorphisms of sheaves:
\be
\Conf_\pi\simeq \bConf_\cM\, , \quad \Sol_{\pi,\bPhi}\simeq \bSol_{\cM,\cG,\Phi}^{g}\, .
\ee
which present $\bConf_\cM$ and $\bSol_{\cM,\cG,\Phi}$ respectively as
the sheaves of configurations and solutions of the section sigma model
defined by the scalar structure $(\pi,\bPhi)$. These isomorphisms of
sheaves encode the {\em U-fold interpretation} of the section sigma
model defined by $(\pi,\bPhi)$.

\subsection{Classical scalar locally-geometric U-folds}

Let $(\cM,\cG,\Phi)$ be a scalar structure. The previous discussion
motivates the following mathematically rigorous definition:
	
\begin{definition}
A {\em classical scalar locally-geometric U-fold} of type
$(\cM,\cG,\Phi)$ is a smooth global solution $s\in \Gamma(\pi)$ of the
equations of motion \eqref{sceom} of the section sigma model defined
by an {\em integrable} bundle of scalar data ($\pi:(E,h)\rightarrow
(M,g),\bPhi)$ having type $(\cM,\cG,\Phi)$.
\end{definition}

\noindent
As explained above, any such object can be constructed by gluing local
solutions of the ordinary sigma model with target $(\cM,\cG)$ using
the $\Iso(\cM,\cG,\Phi)$-valued transition functions of
$\pi$. Moreover, the discussion above shows that a section sigma model
based on an integrable Kaluza-Klein space is {\em locally
indistinguishable} from an ordinary sigma model. Thus {\em section
sigma models based on integrable Kaluza-Klein spaces provide allowed
globalizations of the local formulas used in the sigma model
literature}, in the sense that they are locally indistinguishable from
the latter. When the space-time $M$ is not simply-connected, the
number of inequivalent global formulations of this type is in general
{\em continuously infinite}, since so is the character variety of
$\pi_1(M)$ for the group $\Iso(\cM,\cG,\Phi)$. This illustrates the
highly ambiguous character of the local formulation of theories
involving sigma models (such as supergravity theories coupled to
scalar matter in four dimensions). Clearly such local formulations are
far from sufficient when one tries to specify the theory uniquely on a
non-contractible space-time.
	
\section{Scalar-electromagnetic bundles}
\label{sec:scalarelectro}
	
Let $\pi:(E,h)\rightarrow (M,g)$ be a Kaluza-Klein space with
Ehresmann transport $T$ associated to the horizontal distribution
$H\subset TE$. Let $\bDelta = (\bcS,\bomega,\bD)$ be a flat symplectic
vector bundle defined over $E$ with symplectic structure $\bomega$ and
symplectic connection $\bD$. Given a point $m\in M$, let
$(\cS_m,D_m,\omega_m)$ be the restriction of $(\bcS,\bomega,\bD)$ to
the fiber $E_m$. This is a flat symplectic vector bundle defined on
the Riemannian manifold $(E_m,h_m)$ and hence a duality structure as
defined in \cite{GESM}. For any path $\Gamma\in\cP(E)$ in the total
space of $E$, let $U_\Gamma:\bcS_{\Gamma(0)}\rightarrow
\bcS_{\Gamma(1)}$ be the parallel transport defined by $\bD$ along
$\Gamma$. Since $\bD$ is a symplectic connection, $U_\Gamma$ is a
symplectomorphism between the symplectic vector spaces
$(\bcS_{\Gamma(0)}, \bomega_{\Gamma(0)})$ and
$(\bcS_{\Gamma(1)},\bomega_{\Gamma(1)})$. For any path $\gamma\in
\cP(M)$, let $\bgamma_e\in \cP(E)$ denote its horizontal lift starting
at the point $e\in E_{\gamma(0)}$. By the definition of $T$, we have
$T_\gamma(e) = \bgamma_e(1)$.
	
\begin{definition}
The {\em extended horizontal transport} along a path
$\gamma\in \cP(M)$ is the unbased isomorphism of vector bundles
$\mT_\gamma:\cS_{\gamma(0)}\rightarrow \cS_{\gamma(1)}$ defined
through:
\be
\mT_\gamma(e)\eqdef U_{\bgamma_e}:\cS_{e}\rightarrow \cS_{T_\gamma(e)}\, , \quad \forall\, e\in E_{\gamma(0)}\, ,
\ee
which linearizes the Ehresmann transport $T_\gamma:E_{\gamma(0)}\rightarrow E_{\gamma(1)}$ along $\gamma$.
\end{definition}

\noindent 
Clearly $\mT_\gamma$ is an isomorphism of flat symplectic vector bundles:
\begin{equation*}
\mT_\gamma\colon (\cS_{\gamma(0)},D_{\gamma(0)},\omega_{\gamma(0)})\xrightarrow{\sim}(\cS_{\gamma(1)},D_{\gamma(1)},\omega_{\gamma(1)})\, ,
\end{equation*}
which lifts the isometry $T_\gamma:(E_{\gamma(0)}, h_{\gamma(0)})\rightarrow	(E_{\gamma(1)}, h_{\gamma(1)})$.

\begin{definition}
Let $\pi\colon (E,h)\to (M,g)$ be a Kaluza-Klein space. A
\emph{duality bundle} $\bDelta$ is a flat symplectic vector bundle
$\bDelta = (\bcS,\bomega,\bD)$ over $E$. Let $\bDelta_{1}$ and
$\bDelta_{2}$ be duality bundles. A \emph{morphism of duality bundles}
from $\bDelta_{1}$ to $\bDelta_{2}$ is a bundle morphism of the
underlying flat symplectic vector bundles.
\end{definition}

\noindent
Let $\bDelta=(\bcS,\bD,\bomega)$ be a duality bundle over $\pi\colon
E\to M$ and let $m_{0}\in M$ be a fixed point in $M$. Since $M$ is
path-connected, it follows that all fiber restrictions
$(\cS_m,D_m,\omega_m)$, $m\in M$, can be recovered by extended
horizontal transport from the flat symplectic vector bundle
$(\cS_{m_0},D_{m_0},\omega_{m_0})$ over $E_{m_0}$. The flat vector
bundle $(\cS_m,D_m,\omega_m)$ is a duality structure as defined in
\cite{GESM}. In particular, the isomorphism class of the duality
structures $\Delta_m\eqdef (\cS_m,D_m,\omega_m)$ is independent of $m$
and is called the {\em type} of $\bDelta$. We will drop the subscript
and write $\Delta\eqdef (\cS,D,\omega)$ for the type of
$\bDelta$. Notice that $\bDelta$ can be viewed as a bundle whose
fibers are the duality structures $\Delta_m$, endowed with the
complete Ehresmann connection given by the extended horizontal
transport $\mT$. Such objects defined over $(M,g)$ form a category
when equipped with the obvious notion of (based) morphism.

\subsection{Vertical tamings and scalar-electromagnetic bundles}

Let $\bDelta$ be a duality bundle over a Kaluza-Klein space
$\pi:(E,h)\rightarrow (M,g)$. Recall that a \emph{taming}
$\bJ\in\Aut(\bcS,\bomega)$ is an automorphism of the symplectic vector
bundle $(\bcS,\bomega)$ satisfying \cite{DuffSalamon}:
\begin{equation*}
\bJ^{2} = -\mathrm{Id}_{\bcS}\, , \qquad \bomega(\bJ e, e) >0\, , \qquad \forall e\in \Gamma(E,\bcS)\, .
\end{equation*}
Tamings always exist. Given a taming $\bJ$ of $(\bcS,\bomega)$ and a
point $m\in M$, we denote by $J_m\eqdef \bJ|_{E_m}$ the taming on
$(\cS_m,\omega_m)$ induced by the restriction of $\bJ$ to the fiber
$E_m$ of $\pi$.
	
\begin{definition}
A taming $\bJ$ of $(\bcS,\bomega,\bD)$ is called {\em vertical} if it
is $\mT$-invariant, which means that it satisfies:
\be
\mT_{\gamma}\circ J_{\gamma(0)}=J_{\gamma(1)}\circ \mT_\gamma\, , \qquad \forall\, \gamma\in \cP(M)\, .
\ee
\end{definition}
	
\noindent 
It is clear that $\bJ$ is vertical if and only if it satisfies: 
\be
\bD_X\circ\bJ=\bJ\circ \bD_X\, ,\qquad \forall\, X\in \Gamma(E,H)\, .
\ee
In this case, $\mT_\gamma$ is an isomorphism of tamed flat symplectic vector bundles:
\begin{equation*}
\mT_\gamma\colon (\cS_{\gamma(0)},\omega_{\gamma(0)},D_{\gamma(0)},J_{\gamma(0)}) \xrightarrow{\sim} (\cS_{\gamma(1)},\omega_{\gamma(1)},D_{\gamma(1)},J_{\gamma(1)})\, ,
\end{equation*}
which covers the isometry
$T_\gamma:(E_{\gamma(0)},h_{\gamma(0)})\rightarrow (E_{\gamma(1)},
h_{\gamma(1)})$, i.e. the following diagram commutes:
\begin{equation*}
\scalebox{1.2}{
	\xymatrix{
(\cS_{\gamma(0)},\omega_{\gamma(0)},D_{\gamma(0)},J_{\gamma(0)}) ~~\ar[d]_{\pi_{\gamma(0)}} \ar[r]^{\!\!\!\!\!\!\mT_{\gamma}} & ~~ (\cS_{\gamma(1)},\omega_{\gamma(1)},D_{\gamma(1)},J_{\gamma(1)}) \ar[d]^{\pi_{\gamma(1)}}\\
(E_{\gamma(0)},h_{\gamma(0)})  \ar[r]^{T_{\gamma}} & (E_{\gamma(1)}, h_{\gamma(1)}) \\
	}} 
\end{equation*}

\begin{definition}
Let $\pi\colon (E,h)\to (M,g)$ be a Kaluza-Klein space. An
\emph{electromagnetic bundle} $\bXi$ is a duality bundle $\bDelta$
over $E$ equipped with a vertical taming $\bJ$. We write $\bXi \eqdef
(\bDelta,\bJ) = (\bcS,\bomega,\bD,\bJ)$. Let $\bXi_{1}$ and $\bXi_{2}$
be two electromagnetic bundles. A morphism of electromagnetic bundles
$f\colon \bXi_{1}\to \bXi_{2}$ from $\bXi_{1}$ to $\bXi_{2}$ is a
morphism of the underlying duality structures which satisfies $\bJ_{2}
\circ f = f\circ \bJ_{1}$.
\end{definition}

\begin{definition}
A {\em scalar-electromagnetic bundle} defined over $(M,g)$ is a triple
$\bcD=(\pi:(E,h)\rightarrow (M,g),\bPhi, \bXi)$ consisting of a
Kaluza-Klein space $\pi:(E,h)\rightarrow (M,g)$, a vertical potential
$\bPhi$ and an electromagnetic bundle $\bXi$ defined over the total
space $E$ of $\pi$. The scalar-electromagnetic bundle $\bcD$ is called
{\em integrable} if $\pi:(E,h)\rightarrow (M,g)$ is an integrable
Kaluza-Klein space.
\end{definition}
	
\noindent 
Let $\bcD=(\pi:(E,h)\rightarrow (M,g), \bPhi, \bXi)$ be a
scalar-electromagnetic bundle, which for simplicity (and when no
confusion can arise) we will denote by $\bcD=(\pi, \bPhi,
\bXi)$. Since $M$ is path-connected, it follows that all fiber
restrictions $(\pi, \bPhi, \bXi)|_{E_{m}} \eqdef
(E_{m},h_{m},\Phi_{m},\cS_m,\omega_m,D_m, J_{m})$ can be recovered by
extended horizontal transport from the 'value' of $\bcD$ at a fixed
point $m_{0}\in M$. In particular, each $\cD_m\eqdef (E_m, h_m,
\Phi_m,\cS_m,\omega_m,D_m,J_m)$ is a scalar-electromagnetic structure
in the sense of \cite{GESM}. The isomorphism class of $\cD_m$ is
independent of $m$ and is called the {\em type} of the
scalar-electromagnetic bundle $\bcD$. Hence we will drop the subscript
and write $\cD \eqdef (\cM, \cG, \Phi,\cS,\omega,D,J)$. Notice that
$\bcD$ can be viewed as a bundle whose fibers are the
scalar-electromagnetic structures $\cD_m$, endowed with the
complete Ehresmann connection  given by the extended horizontal
transport $\mT$. Such objects defined over $(M,g)$ form a category
when equipped with the obvious notion of (based) morphism.

\begin{definition}
The {\em extended holonomy group} of $\bcD$ at the point $m\in M$ is the
subgroup of $\Aut(\cD_m)$ defined through:
\be
\mG_m\eqdef \{\mT_{\gamma}\,|\, \gamma\in \cP(M) , \gamma(0)=\gamma(1)=m\}\subset \Aut(\cD_m)\, .
\ee
\end{definition}

\begin{remark}
The holonomy groups associated to $\bcD$ at different points in $M$
are related by conjugation inside $\Aut(\cD)$ and hence
isomorphic. Therefore, we can speak of the holonomy group $\mG$ of
$\bcD$ without further reference to base points.
\end{remark}

\subsection{Topologically trivial scalar-electromagnetic bundles}
	
Let $\pi^0:(E^0,h^{0})\rightarrow (M,g)$ be a topologically trivial
Kaluza-Klein space over $(M,g)$ with fiber $(\cM,\cG)$, where
$E^0=M\times \cM$ and $\pi^0$ is the projection in the first
factor. Let $p^0$ be the projection of $E^0$ on the second factor. Let
$H\subset TE^0$ be the horizontal distribution determined by $h$ and
let $\cD=(\cM,\cG,\Phi,\cS,\omega,D,J)$ be a scalar-electromagnetic
structure.  There exists a unique flat connection $\bD^{0}$ on the
pulled-back bundle $\cS^{p^0}$ which satisfies the following
conditions for all $e\in E^0$:
\beqan
&& \bD^{0}_x(\sigma^{p^0})=0\, , \quad \forall\, \sigma \in \Gamma(\cM,\cS)\, , \quad \forall\, x\in H_e(h)\, ,\nn\\
&& \bD^{0}_x=D^{p^0}_x\, , \quad \forall\, x\in V_{e}\, ,
\eeqan
where $D^{p^0}$ is the pullback of the connection $D$ through
$p^0$. Let $\bcS^0\eqdef S^{p^0}$, $\bJ^0\eqdef J^{p^0}$ and
$\bomega^0\eqdef \omega^{p^0}$. Then $\bXi^{0}\eqdef
(\bcS^0,\bomega^0,\bD^{0},\bJ^0)$ is an electromagnetic bundle defined
over $E^0$ and $\bPhi^0\eqdef \Phi^{p^0}\eqdef \Phi\circ p^0$ is a
vertical potential for $\pi^0$. A section $\bsigma\in
\Gamma(E^0,\bcS^0)$ satisfies $\bD^0_X\bsigma=0$ for all $X\in
\Gamma(E^0,H)$ if and only if there exists $\sigma\in \Gamma(\cM,\cS)$
such that $\bsigma=\sigma^{p^0}$. It is easy to see that $\bsigma$ is
invariant under the extended horizontal transport $\mT^0$ associated
to $H^{0}$ and $\bD^0$ if and only if $\sigma$ is invariant under the
action of the subgroup:
\be
{\hat G}^h\eqdef \{{\hat T}^h_\gamma\,\, |\,\, \gamma\in \cP(M)\}\subset \Isom(\cM,\cG)\, ,
\ee
where ${\hat T}_\gamma$ was defined in \eqref{hatT}. 

\begin{definition}
The triple $\bcD^{0} = (\pi^0:(E^0,h^{0})\rightarrow (M,g),
\bPhi^0,\bXi^{0})$ is called a {\em topologically trivial
  scalar-electromagnetic bundle} of type $\cD$ defined over $(M,g)$.
\end{definition}

\begin{definition}
A topologically-trivial scalar-electromagnetic bundle $\bcD^0$ whose
underlying Kaluza-Klein space $\pi^0:(E^0,h^0)\rightarrow (M,g)$ is
the product Kaluza-Klein bundle (where $h^0=g\times \cG$) is called a
{\em metrically trivial} scalar-electromagnetic bundle defined over
$(M,g)$.
\end{definition}	

\noindent	
For a metrically-trivial scalar-electromagnetic bundle, we have ${\hat
G}=\id_{\cM}$. Hence a section $\bsigma\in \Gamma(E^0,\bcS^0)$ is
$\mT^0$-invariant if and only if $\bsigma=\sigma^{p^0}$ for
some section $\sigma\in \Gamma(\cM,\cS)$ of $\cS$.

\subsection{Special trivializing atlases for scalar-electromagnetic bundles}
\label{spectrivgauge}

Consider a scalar-electromagnetic bundle $\bcD=(\pi, \bPhi,\bXi)$. Let
$(U_\alpha)_{\alpha\in I}$ be a convex cover of $M$ with $0\not \in I$. Fix
points $m_0\in M$, $m_{\alpha}\in U_\alpha$ as well as paths
$\lambda^{\alpha}$ from $m_{\alpha}$ to $m_0$ as in Subsection
\ref{spectriv}, whose notations we will use freely. Let $\cD\eqdef
(\cM,\cG,\Phi,\cS,\omega,D,J)$ be the type of a scalar-electromagnetic
bundle $\bcD$ and let $\bcD_\alpha \eqdef (\pi_\alpha,
\bPhi_\alpha,\bXi_{\alpha})$ denote its restriction to
$E_\alpha$. Then $\bcD_\alpha \eqdef (\pi_\alpha,
\bPhi_\alpha,\bXi_{\alpha})$ is a scalar-electromagnetic bundle
defined over $(E_{\alpha},h_{\alpha})$, with Kaluza-Klein metric given
by $h_{\alpha}\eqdef h|_{U_{\alpha}}$. Let $\bcD^0_{\alpha}=
(\pi^{0}_{\alpha}, \bPhi^{0}_{\alpha},\bXi^{0}_{\alpha})$ be the
topologically trivial scalar-electromagnetic bundle of type $\cD$
defined over $E^{0}_\alpha$, with Kaluza-Klein metric $h^{0}_{\alpha}$
making:
\begin{equation*}
q_{\alpha}\colon E_{\alpha}\to U_{\alpha}\times \cM\, ,
\end{equation*}
into an isometry (see Subsection \ref{spectriv}). We have
$\bcS^0_{\alpha}=\cS^{p^0_{\alpha}}$, $\bD^0_{\alpha}=D^{p^0_{\alpha}}$,
$\bJ^0_{\alpha}=J^{p^0_{\alpha}}$, $\bomega^0_{\alpha}=\omega^{p^0_{\alpha}}$.

The extended horizontal transport $\mT$ along geodesics inside
$U_\alpha$ can be used to define unbased isomorphisms of
electromagnetic structures
$\mq_\alpha:(\bcS_\alpha,\bomega_\alpha,\bD_\alpha,\bJ_\alpha)\rightarrow
(\bcS^0_{\alpha}, \bomega^0_{\alpha},\bD^0_{\alpha}, \bJ^0_{\alpha})$
which linearize the diffeomorphisms $q_\alpha:E_\alpha\rightarrow
E^0_{\alpha}$ defined in equation \eqref{qalpha}. For any $e\in
E_\alpha$, the linear isomorphism
$\mq_{\alpha,e}:\cS_e\stackrel{\sim}{\rightarrow}
\cS^{p^0_{\alpha}}_{q_\alpha(e)}=\cS_{T_{\lambda^{\alpha}\circ\gamma_{\pi(e)}^{\alpha}}(e)}$
is defined through:
\be
\mq_{\alpha,e}=\mT_{\lambda^{\alpha}\circ\gamma_{\pi(e)}^{\alpha}}(e)\, .
\ee
This gives commutative diagrams:
\ben
\label{ltrivextended}
\scalebox{1.2}{
\xymatrix{
			\bcS_\alpha
			~~\ar[d] \ar[r]^{\!\!\!\!\!\!\mq_\alpha} &
			~~\bcS^0_{\alpha} \ar[d]\\ E_\alpha
			~~\ar[d]_{\pi_\alpha} \ar[r]^{\!\!\!\!\!\!q_\alpha}
			& ~~E^0_{\alpha} \ar[d]^{\pi^0_{\alpha}}\\
			U_\alpha \ar[r]^{\id_{U_\alpha}} & U_\alpha \\
			}}
\een
which extend the diagrams \eqref{ltriv}. In particular, $q_\alpha$ are
determined by $\mq_\alpha$. It is easy to see that $\mq_\alpha$ is an
isomorphism of electromagnetic structures from
$(\bcS_\alpha,\bomega_\alpha,\bD_\alpha,\bJ_\alpha)$ to
$(\bcS^0_{\alpha},\bomega^0_{\alpha},\bD^0_{\alpha},
\bJ^0_{\alpha})$. Hence $\mq_\alpha$ can be viewed as an isomorphism
of scalar-electromagnetic structures from $\bcD_\alpha$ to
$\bcD_{\alpha}^{0}$.  By definition, the family
$(U_\alpha,\mq_\alpha)_{\alpha\in I}$ is the {\em special trivializing
  atlas} of the scalar-electromagnetic bundle $\bcD$ defined by the
convex cover $(U_\alpha)_{\alpha\in I}$ and by the data $m_0$ and
$(m^{\alpha},\lambda^{\alpha})_{\alpha\in I}$.  Since any
scalar-electromagnetic bundle admits such an atlas, it follows that
any scalar-electromagnetic bundle is locally isomorphic with a
topologically trivial scalar-electromagnetic bundle.

For any $\alpha,\beta\in I$ such that $U_{\alpha\beta}\neq \emptyset$
and any $m\in U_{\alpha\beta}$, the restriction of the unbased
isomorphism $\mq_\beta\circ
\mq^{-1}_\alpha:\bcS^0_{\alpha}|_{E_{\alpha\beta}}\rightarrow
\bcS^0_{\beta}|_{E_{\alpha\beta}}$ to the fiber $\{m\}\times \cM$ can
be identified with the unbased automorphism $\mf_{\alpha\beta}(m)\in
\Aut^\ub(\cS)$ of $\cS$ given by:
\be
\mf_{\alpha\beta}(m)=\mT_{c^{\alpha\beta}_m}\, ,
\ee
where $c^{\alpha\beta}_m$ is the closed path defined in
\eqref{calphabeta}. This unbased automorphism of $\cS$ linearizes the
isometry $\bg_{\alpha\beta}(m)\in \Iso(\cM,\cG,\Phi)$ of $\cM$ given
by the transition function $\bg_{\alpha\beta}$. Since $\bJ$ and
$\bomega$ and the vertical connection induced by $\bD$ are
$\mT$-invariant, we have $\mf_{\alpha\beta}(m)\in
\Aut(\cD)=\Aut^\ub_{\cM,\cG,\Phi}(\cS,D,\omega,J)$. This gives maps
$\mf_{\alpha\beta}:U_{\alpha\beta}\rightarrow \Aut(\cD)$, which we
call the {\em extended transition functions} of the special
trivializing atlas $(U_\alpha,\mq_\alpha)_{\alpha\in I}$. Notice that
$\bg_{\alpha\beta}$ are uniquely determined by $\mf_{\alpha\beta}$. In
addition, notice that $\mf_{\alpha\beta}$ satisfy the cocycle
condition:
\ben
\label{fcocycle}
\mf_{\beta\gamma}\mf_{\alpha\beta}=\mf_{\alpha\gamma}\, ,
\een
which implies the cocycle condition \eqref{gcocycle} for
$\bg_{\alpha\beta}$.
		
When $\bcD$ is integrable, any homotopy between paths in $M$ lifts to
a homotopy between the horizontal lifts of those paths in $E$. This
implies that the transition functions $\mf_{\alpha\beta}$ are constant
on $U_{\alpha\beta}$ and hence they can be viewed as elements of
$\Aut(\cD)$ which cover $\bg_{\alpha\beta}\in \Iso(\cM,\cG,\Phi)$.

\subsection{The fundamental bundle form and field of an electromagnetic bundle}
\label{sec:fundamental}

Consider a scalar-electromagnetic bundle $\bcD = (\pi,\bPhi,\bXi)$
with associated electromagnetic bundle $\bXi = (\bcS,\bomega,\bD,\bJ)$
and Kaluza-Klein space $\pi\colon (E,h)\to (M,g)$. Let:
\begin{equation*}
\bD^{\ad}\colon \Gamma(E,End(\bcS))\to \Omega^{1}(E,End(\bcS))\, ,
\end{equation*}
be the connection induced by $\bD$ on the endomorphism bundle
$End(\bcS)$ of $\bcS$.

\begin{definition}
The \emph{fundamental bundle form} $\bTheta$ of $\bcD$ is
the $End(\bcS)$-valued one-form defined on $E$ as follows:
\begin{equation*}
\bTheta \eqdef \bD^{\ad} \bJ \in \Omega^{1}(E,End(\bcS))\, .
\end{equation*}
\end{definition}

\noindent
The fact $\bJ$ is vertical together with the fact that the
decomposition $TE = H\oplus V$ is $h$-orthogonal implies that we have:
\begin{equation*}
\bTheta \in \Gamma(E, V^{\ast}\otimes End(\bcS))\, .
\end{equation*}

\begin{definition}
The \emph{fundamental bundle field} $\bPsi$ of $\bcD$ is
the $End(\bcS)$-valued vector field defined on $E$ as follows:
\begin{equation*}
\bPsi \eqdef \left(\sharp_{h}\otimes \mathrm{Id}_{End(\bcS)} \right)\circ D^{\ad} \bJ \in \Gamma(E,V\otimes End(\bcS))\, .
\end{equation*}
\end{definition}

\noindent
We denote by $End(\cS_{m}) = End(\bcS)|_{E_{m}}$ the restriction of
$End(\bcS)$ to $E_{m}$, which becomes the endomorphism bundle of the
vector bundle $\cS_{m}$. We denote by $\Theta_{m}\eqdef
\bTheta|_{E_{m}}$ the restriction of $\bTheta$ to $E_{m}$, which is a
section of $V^{\ast}|_{E_{m}}\otimes End(\cS_{m})$, namely:
\begin{equation*}
\Theta_{m} \in \Gamma(E_{m}, V^{\ast}|_{E_{m}}\otimes End(\cS_{m}))\, .
\end{equation*}
Likewise, we define $\Psi_{m}\eqdef \bPsi|_{E_{m}}$, which is a
section of $V|_{E_{m}}\otimes End(\cS_{m})$:
\begin{equation*}
\Psi_{m} \in \Gamma(E_{m}, V|_{E_{m}}\otimes End(\cS_{m}))
\end{equation*}
Since $V\subset TE$ is the vertical integrable distribution integrated
by the fibers of $\pi\colon (E,h)\to (M,g)$ (which by assumption are
connected), we have $V|_{E_{m}} \simeq TE_{m}$ and we obtain:
\begin{equation*}
\Theta_{m} = \Omega^{1}(E_{m}, End(\cS_{m}))\, , \qquad \Psi_{m}\in \mathfrak{X}(E_{m},End(\cS_{m}))\, .
\end{equation*}

\noindent
The extended horizontal transport $\mT_{\gamma}$ along paths
$\gamma\in \cP(M)$ induces various isomorphisms of spaces of sections
of the appropriate vector bundles. For simplicity we denote all these
isomorphisms by the same symbol. For instance, $\mT_{\gamma}$ induces
the following isomorphism:
\begin{equation*}
\mT_{\gamma}\colon \Omega^{1}(E_{\gamma(0)}, End(\cS_{\gamma(0)}))\xrightarrow{\sim} \Omega^{1}(E_{\gamma(1)}, End(\cS_{\gamma(1)}))\, ,
\end{equation*}
whose explicit action on homogeneous elements $\alpha\otimes \Sigma\in
\Omega^{1}(E_{\gamma(0)}, End(\cS_{\gamma(0)}))$, with $\alpha\in
\Omega^{1}(E_{\gamma(0)})$ and $\Sigma\in \Gamma(E_{\gamma(0)},
End(\cS_{\gamma(0)}))$ is given by:
\begin{equation*}
\mT_{\gamma}\cdot \left( \alpha\otimes \Sigma\right) = (T^{-1}_{\gamma})^{\ast}\alpha \otimes \left( \mT_{\gamma}\cdot \Sigma\right)\, .
\end{equation*}
Here $\mT_{\gamma}\cdot \Sigma\in
\Gamma(E_{\gamma(1)},End(\cS_{\gamma(1)}))$. The explicit
evaluation of $\mT_{\gamma}\cdot \Sigma\colon \cS_{\gamma(1)} \to
\cS_{\gamma(1)}$ on sections of $\cS_{\gamma(1)}$ takes the form:
\begin{equation*}
(\mT_{\gamma}\cdot \Sigma) (\xi) = \mT_{\gamma}\cdot \Sigma (\mT_{\gamma}^{-1}\cdot \xi) = 
\mT_{\gamma}\cdot \Sigma (\mT_{\gamma}^{-1}\circ \xi\circ T_{\gamma}) =  \mT_{\gamma}\circ \Sigma \circ \mT_{\gamma}^{-1}\circ \xi\, , \qquad \xi \in \Gamma(E_{\gamma(1)}, \cS_{\gamma(1)})\, .
\end{equation*}
As a general rule, and for simplicity in the notation, we will denote
the action induced by extended horizontal transport $\mT_{\gamma}$ on
a given module of sections by $``\cdot"$, whereas we will denote by
$``\circ"$ the action induced by $\mT_{\gamma}$ as an unbased
automorphism of the given bundles\footnote{Unfortunately, we cannot
  use the same notation as in \cite{GESM} since bold letters have in
  this article a different meaning.}. The explicit form of the action
represented by $``\cdot"$ will depend on the particular details of the
module of sections which is acted upon and should be clear from
the context. Likewise, the explicit form of the action represented by
$``\circ"$ will depend on the bundles involved. The reader is referred
to appendix D of \cite{GESM} and to remark \ref{remark:Tgammaction}
for detail about the explicit form of the various actions
induced by $\mT_{\gamma}$.

Let us fix $m_{0}\in M$. The fact that the restriction $\Psi_{m_{0}}$
of the fundamental bundle field $\bPsi$ to $E_{m_{0}}$ is a vector
field over $E_{m_{0}}$ taking values on $End(\cS_{m_{0}})$ together
with the fact that the extended horizontal transport $\mT_{\gamma}$
preserves the flat symplectic connection on $\bcS$ and covers
isometries, implies:
\begin{equation}
\label{eq:isoTheta}
\mT_{\gamma} \cdot \Theta_{m_{0}} = \Theta_{m}\, , \qquad \forall\, m\in M\, ,
\end{equation}
\begin{equation}
\label{eq:isoPsi}
\mT_{\gamma} \cdot\Psi_{m_{0}} = \Psi_{m}\, , \qquad \forall\, m\in M\, ,
\end{equation}
where $\gamma\in\cP(M)$ is a path in $M$ with initial point $\gamma(0)
= m_{0}$ and final point $\gamma(1) = m$. We conclude that the
isomorphism type of the restriction of the fundamental bundle form
$\bTheta$ and the fundamental bundle field $\bPsi$ to $E_{m}$ does not
depend on the point and hence will be denoted by $\Theta$ and $\Psi$,
respectively.

\begin{remark}
\label{remark:Tgammaction}
Equation \eqref{eq:isoTheta} follows from the following
computation\footnote{Equation \eqref{eq:isoPsi} can be proved
  similarly.}:
\begin{equation*}
\mT_{\gamma} \cdot \Theta_{m_{0}} = \mT_{\gamma}\cdot D^{\ad}_{m_{0}}(\mT^{-1}_{\gamma}\cdot\mT_{\gamma}\cdot J_{m_{0}}) = \mT_{\gamma}\cdot D^{\ad}_{m_{0}}(\mT^{-1}_{\gamma}\cdot J_{m}) = D^{\ad}_{m} J_{m} = \Theta_{m}\, ,
\end{equation*}
where we have used the fact that $\mT_{\gamma}$ preserves $D_{m_{0}}$ and $D_{m}$, i.e.:
\begin{equation}
\label{eq:D0m}
D^{\ad}_{m}(\xi)  =  \mT_{\gamma}\cdot D^{\ad}_{m_{0}} (\mT^{-1}_{\gamma}\cdot\xi) = \mT_{\gamma}\cdot D^{\ad}_{m_{0}} (\mT^{-1}_{\gamma}\circ\xi \circ T_{\gamma})\, , \quad\forall\,\, \xi\in\Gamma(E_{m},\cS_{m})\, ,
\end{equation}
as well as the relations:
\begin{eqnarray}
\label{eq:J0m}
& (\mT_{\gamma}\cdot J_{m_{0}})(\xi) = \mT_{\gamma}\cdot J_{m_{0}}(\mT^{-1}_{\gamma}\cdot\xi) = \mT_{\gamma}\cdot (J_{m_{0}} \circ\mT^{-1}_{\gamma}\circ \xi \circ T_{\gamma}) \nonumber\\
& = \mT_{\gamma}\circ (J_{m_{0}} \circ\mT^{-1}_{\gamma}\circ \xi \circ T_{\gamma})\circ T^{-1}_{\gamma} = \mT_{\gamma}\circ J_{m_{0}} \circ\mT^{-1}_{\gamma}\circ \xi = J_{m}(\xi)\, ,
\end{eqnarray}
which hold for all $\xi\in \Gamma(E_{m},\cS_{m})$. Notice that the
$\mT_{\gamma}$-action symbols $``\cdot"$ and $``\circ"$ in equations
\eqref{eq:D0m} and \eqref{eq:J0m} have a different meaning depending
on the step and the sections involved. In equation \eqref{eq:D0m} we
have:
\begin{equation*}
D^{\ad}_{m_{0}} (\mT^{-1}_{\gamma}\circ\xi \circ T_{\gamma}) \in \Omega^{1}(E_{m_{0}},\cS_{m_{0}})
\end{equation*}
and evaluating on a vector field $X\in\mathfrak{X}(E_{m})$ we obtain:
\begin{equation*}
\iota_{X} D^{\ad}_{m}(\xi) = \mT_{\gamma}\circ \left(\iota_{T^{-1}_{\gamma\ast}X} D^{\ad}_{m_{0}} (\mT^{-1}_{\gamma}\circ\xi \circ T_{\gamma})\right)\circ T^{-1}_{\gamma}\in \Gamma(E_{m},\cS_{m})\, .
\end{equation*}
Likewise, in the second step of equation \eqref{eq:J0m} we use:
\begin{equation*}
\mT^{-1}_{\gamma}\cdot\xi = \mT^{-1}_{\gamma}\circ \xi \circ T_{\gamma}\in \Gamma(E_{m_{0}},\cS_{m_{0}}) \, , \quad \xi\in \Gamma(E_{m},\cS_{m})\, ,
\end{equation*}
where $\mT^{-1}_{\gamma}$ in the left hand side acts through
$``\cdot"$ as the isomorphism of modules:
\begin{equation*}
\mT^{-1}_{\gamma}\colon \Gamma(E_{m},\cS_m) \to \Gamma(E_{m_{0}},\cS_{m_{0}})\, ,
\end{equation*}
whereas $\mT^{-1}_{\gamma}$ in the right-hand-side acts through $``\circ"$
as composition with the unbased automorphism
$\mT^{-1}_{\gamma}\colon \cS_{m}\to \cS_{m_{0}}$ of vector bundles
covering $T^{-1}_{\gamma}$. 
\end{remark}

\noindent
The following proposition follows from the previous discussion and
summarizes the isomorphism type of a fundamental bundle field
$\bPsi$. Similar remarks apply for the fundamental bundle form
$\bTheta$.
\begin{proposition}
Let $\bcD = (\pi,\bPhi,\bXi)$ be a scalar-electromagnetic structure of
type:
\begin{equation*}
\cD = (\cM, \cG, \Phi,\cS,D,\omega,J)
\end{equation*}
Then the isomorphism type of the fundamental bundle form $\bTheta$
of $\bcD$ is given by:
\begin{equation*}
\Theta \eqdef \Theta_{m_{0}} = D^{\ad}J \in \Omega^{1}(\cM,End(\cS))~~,
\end{equation*}
while the isomorphism type of the fundamental bundle field $\bPsi$
of $\bcD$ is given by:
\begin{equation*}
\Psi \eqdef \Psi_{m_{0}} = \left(\sharp_{\cG}\otimes \mathrm{Id}_{End(\cS)} \right) D^{\ad}J \in \mathfrak{X}(\cM,End(\cS)) \, .
\end{equation*}
\end{proposition}

\noindent
Now let $s\in \Gamma(\pi)$ be a section of $\pi\colon (E,h)\to
(M,g)$. The pull-backs through $s$ of $\bPhi$ and $\bPsi$ satisfy:
\begin{equation*}
\bPhi^{s}\in \Gamma(M,(V^{\ast})^{s}\otimes End(\bcS^{s}))\, , \qquad \bPsi^{s}\in \Gamma(E,V^{s}\otimes End(\bcS^{s}))\, .
\end{equation*}
These objects will be used in section \ref{sec:GESM} in the formulation of GESM theories.

\section{Generalized Einstein-Section-Maxwell theories}
\label{sec:GESM}

In this section, we define Generalized Einstein-Section-Maxwell
theories, or GESM theories for short, in terms of a set of partial
differential equations globally formulated on $(M,g)$. A GESM theory
is a generalization of the generalized Einstein-Scalar-Maxwell theory
introduced in reference \cite{GESM}, obtained by promoting the
standard sigma model of the latter to a section sigma model in which
the scalar map is replaced by a section of the corresponding
Lorentzian submersion. When the horizontal distribution $H\subset TE$
of the section sigma model is integrable, GESM theories can be used to
systematically generalize in a globally nontrivial way the bosonic
sector of four-dimensional ungauged supergravity theories.

\subsection{The polarization condition and the electromagnetic equation}

Before defining the complete GESM theory, we consider in this subsection
a truncated version consisting of a section sigma model coupled to a
generically non-trivial duality bundle. We call such theory a
{\em Generalized Section-Maxwell theory}, or GSM theory for short.

\begin{definition}
Let $\bcD=(\pi,\bPhi, \bXi)$ be a scalar-electromagnetic bundle
with associated Kaluza-Klein space $\pi\colon (E,h)\to
(M,g)$ and let $s\in \Gamma(\pi)$ be a section of $\pi$. An {\em
electromagnetic field strength} is a two-form
$\bcV\in \Omega^2(M,\bcS^s)$ having the following properties:
\begin{enumerate}[1.]
\itemsep 0.0em
\item $\bcV$ is positively-polarized with respect to $\bJ^s$, i.e. the
  following relation is satisfied:
\be
\ast_g \bcV=-\bJ^s\bcV\, .
\ee
\item $\bcV$ satisfies the {\em electromagnetic equation} with respect
  to $s$:
\ben
\label{GaugeEOM}
\dd_{\bD^s} \bcV=0~~.
\een
\end{enumerate} 
\end{definition}		
		
\noindent 
We denote by $\Omega^{2+,s}_{g,\bcS,\bJ}$ the sheaf of $\bcS^s$-valued
two-forms which are positively-polarized with respect to $\bJ^s$, thus
$\Omega^{2+,s}_{g,\bcS,\bJ}(M)$ denotes the space of $\bcS^s$-valued
two-forms on $M$ which are positively-polarized with respect to
$\bJ^s$.

\begin{definition}
The \emph{sheaf of electromagnetic field strengths} is the sheaf
$\cE_{g,\bD,\bJ}^s$ of vector spaces defined as follows:
\be
\cE_{g,\bD,\bJ}^s(U) \eqdef \{\bcV\in \Omega^{2+,s}_{g,\bcS,\bJ}(U)\, , \,\, | \,\, \dd_{\bD^s}\bcV=0\}\, ,
\ee
for every open set $U\subset M$. The vector space of global
electromagnetic field strengths in $M$ is then given by
$\cE_{g,\bD,\bJ}^s(M)$.
\end{definition}

\begin{definition}
Let $\bcD = (\pi,\bPhi,\bXi)$ be a scalar-electromagnetic bundle with
associated Kaluza-Klein space $\pi\colon (E,h)\to (M,g)$. The {\em
configuration sheaf} $\Conf_{\bcD}$ of a GSM theory associated to
$\bcD$ is the sheaf of sets defined as follows:
\be
\bConf_{\bcD}(U)\eqdef \left\{ (s,\bcV)\, , \,\, |\,\, s\in \Gamma(\pi|_{U}) \, , \, \bcV\in \Omega^{2+,s}_{g,\bcS,\bJ}(U) \right\}\, ,
\ee
for every open set $U\subset M$.
\end{definition}

\begin{definition}
Let $\bcD = (\pi,\bPhi,\bXi)$ be a scalar-electromagnetic bundle with
associated Kaluza-Klein space $\pi\colon (E,h)\to (M,g)$. The
GSM theory associated to $\bcD$ is defined by the following set of
partial differential equations on $(M,g)$\footnote{The meaning of the
symbol $(-,-)$ will be explained in a moment, see definition \ref{def:tep}.}:
\begin{equation*}
\cE_{S}(s,\bcV) \eqdef \tau^{v}(h,s) + (\grad_{h} \bPhi)^{s}- \frac{1}{2} (\ast \bcV , \bPsi^s\bcV) = 0\, , \quad \cE_{K}(s,\bcV)\eqdef \dd_{\bD^s} \bcV=0\, ,
\end{equation*}
with unknowns given by pairs $(s,\bcV)\in \Conf_{\bcD}(M)$.
\end{definition}

\begin{definition}
The {\em solution sheaf} of a GSM theory associated to a
scalar-electromagnetic bundle $\bcD = (\pi,\bPhi,\bXi)$ is the sheaf
of sets is given by:
\be
\bSol^g_{\bcD}(U)\eqdef \{(s,\bcV)\in \Conf_{\bcD}(U)\, , \,\, |\,\, \cE_{S}(s,\bcV) = 0\, , \quad  \cE_{K}(s,\bcV) = 0\}\, ,
\ee
for every open set $U\subset M$.
\end{definition}

\noindent
GSM theories generalize in a non-trivial way four-dimensional
Maxwell-theory not only by coupling the theory to a section-sigma model
instead of a standard sigma model, but also by coupling
the \emph{field strengths} appearing in the standard formulation of
the theory to a generically non-trivial flat symplectic vector
bundle. Many interesting aspects of these theories, such as their
quantization and existence of solutions, remain to be explored
even for the case when the theory is not coupled to a section sigma
model. 

\subsection{The complete GESM theory}
\label{sec:completetheoryII}

We are now ready to define GESM theories. In physics terms, a
GESM-theory is a theory of gravity coupled to an arbitrary number of
scalars and an arbitrary number of gauge fields with a potential which
depends exclusively on the scalars.  The formulation given below is
appropriate for non-contractible space-times and should be sufficiently
rich to accommodate the classical limit of string theory U-folds when
the GESM theory is considered as a global formulation of the bosonic
sector of four-dimensional supergravity on such space-times.

\begin{definition}
Let $\bcD = (\pi,\bPhi,\bXi)$ be a scalar-electromagnetic bundle with
associated Lorentzian submersion $\pi\colon (E,h)\to (M,g)$. The {\em
  configuration sheaf} of a GESM-theory defined by $\bcD$ is the sheaf
of sets given by:
\be
\bConf_{\bcD}(U)\eqdef \left\{ (g,s,\bcV)\, , \,\, |\,\, g\in\Met_{3,1}(U)\, , \,   s\in \Gamma(\pi|_{U}) \, , \, \bcV\in \Omega^{2+,s}_{g,\bcS,\bJ}(U) \right\}\, ,
\ee
for every open set $U\subset M$.
\end{definition}

\noindent
The {\em exterior pairing} $(~,~)_g$ is the pseudo-Euclidean metric
induced by $g$ on the exterior bundle $\wedge_M$.

\begin{definition}
\label{def:tep}
The {\em twisted exterior pairing} $(~,~):=(~,~)_{g,\bQ^{s}}$ is the
unique pseudo-Euclidean scalar product on the twisted exterior bundle
$\wedge_M(\bcS^{s})$ which satisfies:
\begin{equation*}
(\rho_1\otimes \xi_1,\rho_2\otimes \xi_2)_{g,\bQ^{s}}=(\rho_1,\rho_2)_g \bQ^{s}(\xi_1,\xi_2)\, ,
\end{equation*}
for any $\rho_1,\rho_2\in \Omega^{\bullet}(M)$ and any
$\xi_1,\xi_2\in \Gamma(M,\bcS^{s})$. Here $\bQ^{s}(\xi_1,\xi_2)
= \bomega^{s}(\bJ^{s} \xi_1,\xi_2)$ and the superscript denotes pull-back by
$s$.
\end{definition}

\noindent For any vector bundle $W$, we trivially extend the twisted
exterior pairing to a $W$-valued pairing (which for simplicity we
denote by the same symbol) between the bundles $W\otimes
(\wedge_M(\bcS^{s}))$ and $\wedge_M(\bcS^{s})$. Thus:
\be
(e \otimes \eta_1,\eta_2)_{g,\bQ^{s}} \eqdef  e \otimes (\eta_1,\eta_2)_{g,\bQ^{s}}\, , \quad\forall\, e\in \Gamma(M,W)\, , \quad \forall\, \eta_1,\eta_2\in \wedge_M(\bcS^{s})\, .
\ee

\noindent 
The {\em inner $g$-contraction of 2-tensors} is the bundle morphism
$\oslash_g:(\otimes^2T^\ast M)^{\otimes 2}\rightarrow\otimes^2 T^\ast
M$ uniquely determined by the condition:
\begin{equation*}
(\alpha_1\otimes\alpha_2)\oslash_g (\alpha_3\otimes \alpha_4)=(\alpha_2,\alpha_3)_g\alpha_1\otimes \alpha_4\, , \quad \forall\, \alpha_1, \alpha_2, \alpha_3, \alpha_4\in \Omega^1(M)\, .
\end{equation*}
We define the \emph{inner $g$-contraction of 2-forms} to be the
restriction of $\oslash_g$ to $\wedge^2 T^\ast M \otimes \wedge^2
T^\ast M\stackrel{\oslash_g}{\rightarrow} \otimes^2 T^\ast M$.

\begin{definition}
The {\em twisted inner contraction} of $\bcS^{s}$-valued 2-forms is
the unique morphism of vector bundles
$\loslash\colon\wedge_M^2(\bcS^{s})\times_M\wedge_M^2(\bcS^{s})\rightarrow\otimes^2(T^\ast
M)$ which satisfies the condition:
\begin{equation*}
(\rho_1\otimes \xi_1)\loslash (\rho_2\otimes \xi_2)= \bQ^{s} (\xi_1,\xi_2) \rho_1 \oslash_g\rho_2\, ,
\end{equation*}
for all $\rho_1,\rho_2\in \Omega^2(M)$ and all $\xi_2,\xi_2\in \Gamma(M,\bcS^{s})$. 
\end{definition}

\begin{definition}
Let $\bcD = (\pi,\bPhi,\bXi)$ be a scalar-electromagnetic bundle with
associated Kaluza-Klein space $\pi\colon (E,h)\to (M,g)$. The 
GESM theory associated to $\bcD$ is defined by the following set of
partial differential equations on $(M,g)$:
\begin{itemize}
\item The Einstein equations\footnote{We denote Einstein's gravitational constant by $\kappa$.}:	
\begin{equation}
\cE_{E}(g,s,\bcV) \eqdef \G(g) -\kappa\,\mathrm{T}(g,s,\bcV) = 0\, ,
\end{equation}
where $\mathrm{T}(g,s,\bcV)\in\Gamma(M,S^{2}T^{\ast}M)$ is the
energy-momentum tensor of the theory, which is given by:
\begin{equation*}
\mathrm{T}(g,s,\bcV)  = g\, e^{v}_{0}(g,h,s) - h^{s}_{V} + \frac{g}{2} \bPhi^{s}  + 2~\bcV \loslash \bcV \, .
\end{equation*} 

\item The scalar equations:
\ben
\label{sc}
\cE_{S}(g,s,\bcV) \eqdef \tau^{v}(g,h,s) + (\grad_{h} \bPhi)^{s}- \frac{1}{2} (\ast \bcV , \bPsi^s\bcV) = 0\, .
\een
\item The electromagnetic (or Maxwell) equations:
\begin{equation}
\cE_{K}(g,s,\bcV) \eqdef\dd_{\bD^s} \bcV=0\, .
\end{equation}	

\end{itemize}

\noindent
with unknowns given by triples $(g,s,\bcV)\in \Conf_{\bcD}(M)$.
\end{definition}		

\begin{definition}
The {\em solution sheaf} $\Sol_{\bcD}$ of a GESM theory associated to
a scalar-electromagnetic bundle $\bcD = (\pi,\bPhi,\bXi)$ is the sheaf
of sets given by:
\be
\bSol_{\bcD}(U)\eqdef \{(g,s,\bcV)\in \Conf_{\bcD}(U) \,\, |\,\, \cE_{E}(g,s,\bcV) = 0\, , \quad  \cE_{S}(g,s,\bcV) = 0\, , \quad \cE_{K}(g,s,\bcV) = 0\}
\ee
for every open set $U\subset M$.
\end{definition}

\noindent For any open set $U\subset M$ and any scalar-electromagnetic structure
$\cD$, let $\Conf_\cD(U)$ and $\Sol_\cD(U)$ denote the configuration
and solution sets of the generalized Einstein-Scalar-Maxwell theory
defined in \cite{GESM}.

\begin{theorem}
\label{thm:equivalence}
Let $\bcD = (\pi,\bPhi,\bXi)$ be an scalar-electromagnetic bundle of
type $\cD=(\cM,\cS,\omega,D,J)$, whose Kaluza-Klein space is
integrable. Then any special trivializing atlas
$(U_\alpha,\mq_\alpha)_{\alpha\in I}$ of $\bcD$ induces bijections:
\begin{equation*}
\mathfrak{C}_\alpha \colon \bConf_{\bcD}(U_{\alpha})\xrightarrow{\sim}  \Conf_{\cD}(U_{\alpha})
\end{equation*}
which restrict bijections:
\begin{equation*}
\mathfrak{C}_\alpha \colon \bSol_{\bcD}(U_{\alpha}) \xrightarrow{\sim} \Sol_{\cD}(U_{\alpha})~~.
\end{equation*}
\end{theorem}

\begin{proof}
Given $(g_\alpha,s_\alpha,\bcV_\alpha)\in \bConf_{\bcD}(U_\alpha)$, we set:
\begin{equation*}
\mathfrak{C}_{\alpha}(g_{\alpha},s_{\alpha},\bcV_{\alpha})\eqdef (g_{\alpha},\varphi^\alpha,\cV^\alpha)~~,
\end{equation*}
where $\varphi^\alpha$ and $\cV_\alpha$ are defined as follows:
\begin{itemize}
\item As explained in Subsection \ref{spectrivgauge}, the special
  trivializing atlas $(U_\alpha,\mq_\alpha)$ is underlined by a
  special trivializing atlas $(U_\alpha,q_\alpha)$ of the fiber bundle
  $\pi:E\rightarrow M$, whose trivializing maps are diffeomorphisms
  $q_\alpha:E_\alpha\eqdef \pi^{-1}(U_\alpha)\rightarrow
  E_\alpha^0=U_\alpha\times \cM$. This allows us to define smooth maps
  $\varphi^{\alpha}\in \cC^\infty(U_\alpha,\cM)$ such that
  $s_\alpha=q_\alpha^{-1}\circ \graph(\varphi^\alpha)$ as in
  Subsection \ref{subsec:sigmaUfold}:
\begin{equation*}
\varphi^{\alpha}\eqdef \hat{q}_{\alpha}\circ s_{\alpha} = \ungraph(q_\alpha\circ s_\alpha)=p^{0}_{\alpha}\circ q_{\alpha}\circ s_{\alpha}~~.
\end{equation*}
\item As explained in Subsection \ref{spectrivgauge}, the maps
  $\mq_\alpha:(\bcS_\alpha, \bomega_\alpha, \bD_\alpha,
  \bJ_\alpha)\stackrel{\sim}{\rightarrow}
  (\bcS^0_{\alpha},\bomega^0_{\alpha}, \bD^0_{\alpha},
  \bJ^0_{\alpha})=(\cS^{p_\alpha^0}, \omega^{p_\alpha^0},
  D^{p_\alpha^0}, J^{p_\alpha^0})$ are unbased isomorphisms of
  electromagnetic structures which linearize the diffeomorphisms
  $q_\alpha$. They induce isomorphisms:
\be
\mq_\alpha^{(s_\alpha)}:(\bcS_\alpha^{s_\alpha}, \bomega_\alpha^{s_\alpha}, \bD_\alpha^{s_\alpha},
\bJ_\alpha^{s_\alpha})\stackrel{\sim}{\rightarrow}
(\cS^{\varphi^{\alpha}}, \omega^{\varphi^{\alpha}}, D^{\varphi^{\alpha}},
J^{\varphi^{\alpha}})\eqdef \cD^{\varphi^\alpha}\, ,
\ee
where we noticed that
$(\bcS^0_{\alpha})^{s_\alpha}=(\cS^{p^0_{\alpha}})^{s_\alpha}=\cS^{\varphi^\alpha}$ and
we used the fact that the horizontal part of $\bD^0_{\alpha}$ is the
trivial flat connection. This allows us to identify $\bcV_\alpha$ with the
$\cS^{\varphi^{\alpha}}$-valued two-form defined on $U_\alpha$
through:
\ben
\label{cVlocdef}
\cV^{\alpha}\eqdef \mq_\alpha^{(s_\alpha)}(\bcV_\alpha)\in \Omega^2(U_\alpha,\cS^{\varphi^{\alpha}})\, .
\een
The relation $J^{\varphi^{\alpha}}\circ
\mq_\alpha^{(s_\alpha)}=\mq_\alpha^{(s_\alpha)} \circ \bJ_\alpha^{s_\alpha}$ implies:
\ben
\label{eomcorresp}
\bcV_\alpha\in \Omega^{2+,s_\alpha}_{g_\alpha, \bcS^{s_\alpha}_{\alpha},\bJ^{s_\alpha}_{\alpha}}(U_\alpha)~~\mathrm{iff}~~ \cV^{\alpha}\in \Omega^{2+,\varphi^{\alpha}}_{g_{\alpha},\cS^{\varphi^{\alpha}},J^{\varphi^{\alpha}}}(U_\alpha)~~
\een
\end{itemize}
\noindent
The fact that $\mathfrak{C}_\alpha$ is bijective follows immediately
from its definition.  The fact that $\mathfrak{C}_{\alpha}$ induces a
bijection between solution sets follows from the equivalences:
\beqan
\label{equiv}
& \cE^{\bcD}_{E}(g_{\alpha},s_{\alpha},\bcV_{\alpha}) = 0\Leftrightarrow \cE^{\cD}_{E}(g_{\alpha},\phi^\alpha,\cV_\alpha)=0~~,~~\cE^{\bcD}_{S}(g_{\alpha},s_{\alpha},\bcV_{\alpha}) = 0\Leftrightarrow 
\cE^{\cD}_{S}(g_{\alpha},\phi^\alpha,\cV_\alpha)=0 \nn \\ 
& \cE^{\bcD}_{K}(g_{\alpha},s_{\alpha},\bcV_{\alpha}) = 0\Leftrightarrow \cE^{\cD}_{K}(g_{\alpha},\phi^\alpha,\cV_\alpha)=0~~.
\eeqan
This follows by a direct but somewhat lengthy computation using
relations \eqref{eomcorresp}, which we will not reproduce here. For
example, it is easy to see that the relation
$\dd_{D^{\varphi^{\alpha}}}\circ\mq_\alpha^{(s_\alpha)}=\mq_\alpha^{(s_\alpha)}\circ\dd_{\bD^{s_\alpha}}$
implies that the electromagnetic equation
$\dd_{\bcD^{s_\alpha}}\bcV_\alpha=0$ is equivalent with the equation
$\dd_{D^{\varphi^{\alpha}}}\cV^{\alpha}=0$.
\end{proof}

\begin{remark}
Together with the results of \cite{GESM}, Theorem
\ref{thm:equivalence} shows that a GESM theory represents an
admissible global extension of the locally-formulated bosonic theories
considered by physicists in the context of supergravity.
\end{remark}

\subsection{U-fold interpretation of globally-defined solutions in the integrable case}
\label{section:Ufoldint}

Let $\bcD=(\pi, \bPhi, \bXi)$ be an {\em integrable}
scalar-electromagnetic bundle and let
$(U_\alpha,\mq_\alpha)_{\alpha\in I}$ be a special trivializing atlas
for $\pi$. Let $s\in \Gamma(\pi)$ be a global section of $\pi$ and
$\cV\in\Omega^2(M,\bcS^s)$ be a two-form defined on $M$ and valued in
the pulled-back bundle $\bcS^s$. Let:
\be
g_\alpha\eqdef g|_{U_\alpha}~~,~~s_\alpha\eqdef s|_{U_\alpha}~~,~~\bcV_\alpha\eqdef \bcV|_{U_\alpha}~~
\ee
and $\mathfrak{C}_{\alpha}(g_{\alpha},s_{\alpha},\bcV_{\alpha})=
(g_{\alpha},\varphi^\alpha,\cV^\alpha)$.  Relations \eqref{cVlocdef}
imply the gluing conditions:
\ben
\label{cVgluing}
\cV^{\beta}|_{U_{\alpha\beta}}=\mf_{\alpha\beta}^s \cV^{\alpha}|_{U_{\alpha\beta}}~~,
\een
which accompany the gluing conditions \eqref{phigluing} for $\varphi^{\alpha}$ 
and the trivial gluing conditions:
\ben
\label{ggluing}
g_\alpha|_{U_{\alpha\beta}}=g_\beta|_{U_{\alpha\beta}}
\een
for $g_\alpha$. When $(g,s,\bcV)\in \Sol_{\bcD}(M)$, we have
$(g_\alpha,s_\alpha,\bcV_\alpha)\in \Sol_{\bcD}(U_\alpha)$ and Theorem
\ref{thm:equivalence} implies
$(g_{\alpha},\varphi^\alpha,\cV^\alpha)\in \Sol_{\cD}(U_\alpha)$.

Conversely, any family $(g_\alpha, \varphi^\alpha,
\cV^{\alpha})_{\alpha\in I}$ with $(g_\alpha, \varphi^\alpha,
\cV^{\alpha})\in \Sol_\cD(U_\alpha)$ which satisfies conditions
\eqref{ggluing}, \eqref{phigluing} and \eqref{cVgluing} is obtained by
restriction to $U_\alpha$ from a uniquely-determined global solution
$(g, s,\bcV)$ of the equations of motion of the GESM theory defined by
the scalar-electromagnetic bundle $\bcD$. Once again, this can be
formulated in sheaf language, a formulation whose details we leave to
the reader. These observations justify the following:
		
\begin{definition}
Let $\cD$ be a scalar-electromagnetic structure. A {\em classical
  locally-geometric U-fold} of type $\cD$ is a global solution
$(g,s,\bcV)\in \Sol_{\bcD}(M)$ of the equations of motion of a GESM
theory defined by a scalar-electromagnetic bundle $\bcD$ of type $\cD$
with integrable Kaluza-Klein space. We say that a locally-geometric
U-fold is {\em trivial} if the corresponding extended holonomy
group $\mathbb{G}$ is the trivial group.
\end{definition}

\section{A simple example}
\label{sec:SS}

In this section, we show that the celebrated Scherk-Schwarz
construction \cite{SS} can be understood as an instance of the
integrable case of the general models considered in this paper. Notice
from the outset that, when referring to the Scherk-Schwarz
construction, we do {\em not} assume the presence of any continuous
isometry group of the scalar manifold and hence that we do {\em not}
gauge any such putative group.

Let $M=\R^3\times \rS^1$, where $\rS^1=\{\sigma=e^{2\pi \i
  \theta}|\theta\in \R\}=\U(1)$ denotes the unit circle and let
$\pi_1:M\rightarrow \R^3$ and $\pi_2:M\rightarrow \rS^1$ be the
canonical projections. Let $1\in \rS^1$ be the point which corresponds
to $\theta=0$ and let $g\in \Met_{3,1}(M)$ be a Lorentzian metric of
the form:
\be
\dd s_g^2=\dd s_{g_3}^2+(2\pi)^2R^2\dd\theta^2~~,
\ee
where $g_3$ is a Lorentzian metric on $\R^3$ and $R$ is a positive
parameter. In this example, we have $\pi_1(M)\simeq \pi_1(\rS^1)\simeq
\Z$. Setting $x=(x^0,x^1,x^2)\in \R^3$, let $C\eqdef \{0\}\times
\rS^1\subset M$ be the circle defined inside $M$ by the equation
$x=0$. Orienting $C$ in the direction of decreasing $\theta$ (the
``clockwise'' orientation), its free homotopy class gives a generator
of $\pi_1(M)$.

Given an $n$-dimensional Riemannian manifold $(\cM,\cG)$, an
integrable Kaluza-Klein space $(E,h)$ over $(M,g)$ with fibers
isometric to $(\cM,\cG)$ is determined by a morphism of groups
$\rho:\Z\rightarrow \Iso(\cM,\cG)$, i.e. by an element $U_C\eqdef
\rho(1)\in \Iso(\cM,\cG)$ which describes the Ehresmann holonomy of
$E$ along $C$. Let $\pi_E:E\rightarrow M$ be the projection of $E$ and
$U$ be its Ehresmann transport.

Consider the restriction of $E$ to $C$, which has total space $F\eqdef
\pi_E^{-1}(C)$ and projection $\pi_F\eqdef \pi_E|_{F}$. For
simplicity, we denote the fiber of $F$ above a point $(0,\sigma)\in C$
by $F_\sigma$. For each point $m=(x,\sigma)\in M$, the Ehresmann
transport along the curve $\gamma_m:[0,1]\rightarrow M$ defined
through $\gamma_m(t)=((1-t)x,\sigma)$ gives an isometry
$U_m:E_m\rightarrow F_\sigma$. These isometries allow us to identify
$E$ with the pull-back bundle $\pi_2^\ast(F)$ by sending any point
$e\in E_m$ to the point $(m,U_m(e))\in \pi_2^\ast(F)_m=\{m\}\times
F_\sigma$.  Moreover, we can identify the fiber $F_1$ with the
Riemannian manifold $(\cM,\cG)$.

For any $\sigma\in \rS^1=\U(1)$, let $\ell_\sigma:[0,1]\rightarrow C$
be the curve in $C$ defined by $\ell_\sigma(t)=(0,e^{\i
  (1-t)\arg(\sigma) })$, where $\arg(\sigma)\in [0,2\pi)$. Notice that
  this ``clockwise''-oriented curve satisfies
  $\ell_\sigma(0)=(0,\sigma)$ and $\ell_\sigma(1)=(0,1)$. The
  concatenation $\ell_\sigma\cdot\gamma_m$ is a curve which connects
  $m$ to the point $(0,1)\in C$. The Ehresmann transport along this
  curve gives an isometry $V_\sigma\circ U_m:E_m\rightarrow F_1$,
  where $V_\sigma:F_\sigma \rightarrow F_1$ is the Ehresmann transport
  along $\ell_\sigma$. Thus sections $s\in \Gamma(\pi_E)$ of $E$
  correspond to maps $\varphi:M\rightarrow F_1=\cM$ defined through
  the relation:
\ben
\label{VU}
\varphi(x,\sigma)\eqdef (V_\sigma \circ U_m) (s(x,\sigma))~~.
\een
Recall that $U_C=\rho(1)$ denotes the Ehresmann holonomy of $(E,h)$
around $C$, where $C$ is oriented ``clockwise''. Identifying the free
fundamental group $\pi_1(M)$ with the based first homotopy group
$\pi_1(M,(0,1))$, we have $U_C=\lim_{\epsilon\rightarrow 0^+}
V_{\sigma_\epsilon}$, where for $\epsilon\in (0,1)$ we defined
$\sigma_\epsilon\eqdef e^{2\pi \i (1-\epsilon)}$. Using this in
\eqref{VU} gives:
\be
\lim_{\epsilon\rightarrow 0^+}\varphi(x,\sigma_\epsilon)=(U_C\circ U_m) (s(x,1))=U_C(\varphi(x,1))~~,
\ee
where noticed that $\varphi(x,1)=U_m (s(x,1))$ since
$V_1=\id_{\cM}$. Thus $\varphi(x,\sigma)$ has a branch cut in $\sigma$
at $\sigma=1$, with monodromy $U_C=\rho(1)$ along
$\rS^1$. Accordingly, $\varphi(x,\sigma)$ can be extended to a
multivalued function of $\sigma$ having this monodromy.

The identification \eqref{VU} shows that the section sigma model
defined by $(M,g)$ and $(E,h)$ can be interpreted as a variant of the
ordinary sigma model with source $(M,g)$ and target $(\cM,\cG)$ which
is obtained by promoting $\varphi$ from an ordinary map to a
multivalued map having monodromy $U_C=\rho(1)$ around
$\rS^1$. Mathematically, this is equivalent with an ordinary smooth
map $\tvarphi:\R^4\rightarrow \cM$ defined on the universal
covering space $\R^4$ of $M$ and which satisfies the quasi-periodicity
condition:
\ben
\label{phiper}
\tvarphi(x,\theta+1)=U_C(\tvarphi(x,\theta))~~.
\een
Hence the section sigma model of this simple example gives a geometric
encoding of the Scherk-Schwarz construction applied to the ordinary
sigma model with source $(M,g)$ and target $(\cM,\cG)$. One can
similarly check that the Abelian vector fields of the full GESM model
can be described in this example through Scherk-Schwarz type monodromy
conditions --- except that, as in
\cite{GESM,GESMNote,GeometricUfolds}, we allow for a twist by a
symplectic vector bundle defined on $(\cM,\cG)$ in the case when $\cM$
is not simply-connected. Unlike the formulation in terms of universal
covering spaces or multivalued functions (which is quite
cumbersome\footnote{Notice that the the fundamental group $\pi_1(M)$
may not even be finitely-generated, since typically the space-time $M$
is not compact.}  when studying dynamics through the equations of
motion, especially in the presence of Abelian gauge fields), the
description provided by the GESM theory is intrinsic.

In view of the above, the physics-oriented reader may choose to think
of the integrable case of the general construction presented in this
paper as a global geometric formulation of the ``ultimate
generalization'' of the Scherk-Schwarz construction to the case of
arbitrary non-contractible space-times, while allowing for ``duality
twists'' of the Abelian vector fields when the typical fiber $\cM$ of
the generalized Kaluza-Klein space $E$ is not simply-connected. While
this interpretation can help make the construction of this paper seem
less unfamiliar, a few features should be kept in mind in order to
prevent misunderstanding. First, in the example above we do {\em not}
assume that the radius $R$ of the circle is small since we are are not
interested in reducing a four-dimensional theory to three dimensions
but in describing a four-dimensional theory on the non-simply
connected space-time $M=\R^3\times \rS^1$. Second, we do {\em not}
gauge any putative continuous isometry group; in fact, our
construction does not require that $(\cM,\cG)$ admit any continuous group of
isometries (in particular, the group $\Iso(\cM,\cG)$ can be
discrete). Due to this fact, the geometric U-folds constructed in this
paper need {\em not} be solutions of gauged supergravity and the
Abelian gauge fields of a GESM theory need {\em not} originate from
``gauging'' of any continuous isometry group.

\section{Conclusions and further directions}
\label{sec:conclusions}

The present paper extends ordinary Einstein-Scalar-Maxwell theory 
in a few directions so it may be useful to summarize our main results:

\begin{enumerate}[1.]
\item We generalized the ordinary scalar sigma model from a theory of
  maps to a theory of sections of a Lorentzian Kaluza-Klein space,
  whose solutions have an interpretation as ``classical U-folds'' in
  the integrable case. This generalized model is what we call the {\em
    section sigma model}.
\item We gave a global mathematical formulation of the
  coupling of the section sigma model to Abelian gauge fields,
  consistent with electromagnetic-duality and without assuming
  topological triviality of space-time, of the fibers of the
  Kaluza-Klein space or of the duality bundle. This produces what we
  call a {\em Generalized Einstein-Section-Maxwell (GESM) theory}, a
  theory that provides a non-trivial global realization of the generic
  bosonic sector of four-dimensional supergravity in a manner which is
  consistent with electromagnetic duality.
\item When the Kaluza-Klein space is integrable, we showed that the
  GESM theory is {\em locally indistinguishable} from the generalized
  ESM theory of \cite{GESM} and that its global solutions correspond
  to {\em classical} locally-geometric U-folds. This can be viewed as
  a rigorous mathematical definition of classical locally-geometric
  U-folds with scalar and Abelian gauge fields in the language of
  global differential geometry.
\item When the Kaluza-Klein space is not integrable, we showed that
  GESM theories are locally equivalent with modified sigma models
  of maps into the fiber coupled to Abelian gauge fields, which
  differ {\em locally} from ordinary ESM theories.
\end{enumerate}
		
\noindent The locally-geometric classical U-folds described
in this paper further generalize those constructed in in
\cite{GESM}. From the perspective of the present paper, the special
situation discussed in op. cit. corresponds to the particular case
when $\pi$ is a topologically trivial fiber bundle whose Ehresmann
connection is the trivial integrable connection but whose typical
fiber $(\cM,\cG)$ is not simply-connected.

The present work opens up various directions for further
research. First, it is natural to consider the supersymmetrization of
our theories. When the Kaluza-Klein space is integrable, this can be
achieved by building fibered and twisted versions of four-dimensional
supergravities coupled to matter, which are locally-indistinguishable
(but globally different) from the latter. Imposing the appropriate
supersymmetry constraints on a GESM theory turns out to be quite subtle. 
When the Kaluza-Klein space is not integrable, our bosonic
theories differ {\em locally} from the bosonic sector of ordinary
four-dimensional supergravities coupled to matter and hence may lead
to new locally-defined supergravity models provided that one can find
appropriate local supersymmetrizations. In this regard, we mention
that all known ``no-go theorems'' regarding uniqueness of (ungauged)
supergravity theories in four dimensions assume that the bosonic
sector is described locally by an {\em ordinary} sigma model coupled to
Abelian gauge fields. This is {\em not} the case if one takes the
bosonic sector to be described by a section sigma model whose
Kaluza-Klein space is not integrable.
		
Our formulation allows classical locally-geometric U-folds to be
approached using the well-developed tools of global differential
geometry and global analysis. This could be used to shed light on the
classification of U-fold solutions, on moduli spaces of such and to
construct new classes of classical locally-geometric U-folds. It would
be interesting to extend our constructions to dimensions different
from four as well as to the case of Euclidean signature and to study
how they may relate to ``dimensional reduction'' on generalized
Kaluza-Klein spaces. It would also be interesting to study various
deformation problems for the structures arising in this construction. As a long-term
project, it would be very interesting to gauge our models and to
relate them to the appropriate extensions of gauged supergravity
theories. The reader should also notice that GESM theories contain a
``twisted'' generalization of the standard Maxwell theory of Abelian
gauge fields on Lorentzian four-manifolds which are not
simply-connected, a theory which may itself be of some interest.

\appendix

\section{Pseudo-Riemannian submersions and Kaluza-Klein metrics}
\label{app:submersions}

\subsection{Horizontal distributions and related notions}

For any surjective submersion $\pi:E\rightarrow M$ with $\dim E>\dim
M$ and connected fibers, let $V \eqdef \ker (\dd \pi)\subset TE $
denote the vertical distribution of $\pi$. The differential of $\pi$
induces an unbased linear epimorphism of vector bundles $TE
\rightarrow TM$ which covers $\pi$. This can be identified with a
based epimorphism $\beta_\pi:TE \rightarrow (TM)^\pi$. We have a short
exact sequence of vector bundles over $E$ and based morphisms of
vector bundles:
\ben
\label{Vseq}
0\longrightarrow V \hookrightarrow T E\stackrel{\beta_\pi}{\longrightarrow} (TM)^\pi\longrightarrow 0\, .
\een
A {\em horizontal distribution for $\pi$} is a distribution $H\subset
TE$ such that $TE=H\oplus V$; giving such a distribution amounts to
giving a splitting of the sequence \eqref{Vseq}. Let $\cH(\pi)$ denote
the set of horizontal distributions for $\pi$. Given a horizontal
distribution $H$, let $P_V$ and $P_H$ denote the complementary
projectors of $TE$ on $V$ and $H$ and define $X_H\eqdef P_H X$ and
$X_V\eqdef P_V X$ for any vector field $X\in \cX(E)$. The differential
$\dd\pi$ induces a based isomorphism of vector bundles $(\dd
\pi)_H\eqdef \beta_\pi|_H:H \stackrel{\sim}{\rightarrow}
(TM)^\pi$. The {\em horizontal lift} of a vector field $X\in \cX(M)$
is the vector field $\overline{X}\in \cX(E)$ defined through:
\be
\overline{X}\eqdef (\dd\pi)_H^{-1}(X^\pi)\in \Gamma(M,H)\, .
\ee
Given a smooth section $s\in \Gamma(\pi)$, let $\dd s:TM\rightarrow
TE$ be the differential of $s$, viewed as an unbased morphism of
vector bundles from $TM$ to $TE$ covering $s$. The {\em vertical
  differential} of $s$ is the unbased morphism of vector bundles
$\dd^v s:TM\rightarrow V$ defined through:
\be 
\dd^v s\eqdef P_V\circ \dd s\, , 
\ee 
which covers $s$. Notice that $\dd^v s$ can also be viewed as a
$V^s$-valued one-form $\widehat{\dd^v s} \in\Hom(TM,V^s)\simeq
\Gamma(M,T^\ast M\otimes V^s)$. When no confusion is possible, we will
denote $\widehat{\dd^v s}$ simply by $\dd^{v} s$.

A {\em vertical tensor field defined on $E$} is a section $t\in
\Gamma(E,\otimes^k V^\ast)$, which can also be viewed as an element $t
\in \Hom(\otimes^k V,\R_E)$, where $\R_E$ denotes the trivial real
line bundle defined on the total space of $E$. The {\em vertical
  differential pullback} of $t$ through a section $s\in \Gamma(\pi)$
of $\pi$ is the tensor field defined on $M$ through:
\be
s^\ast_v(t)\eqdef t^s\circ (\otimes^k \widehat{\dd^v s})\in \Hom(\otimes^k TM, \R_M)\simeq \Gamma(M,\otimes^k T^\ast M)\, ,
\ee
where we noticed that $\R_E^s=\R_M$.

Any connection $\nabla$ on $T E$ induces connections
$\nabla^V=P_V\circ \nabla$ on $V$ and $\nabla^H=P^H\circ \nabla$ on
$H$, called the {\em vertical and horizontal connections} induced by
$\nabla$ relative to $H$.

\subsection{Vertically non-degenerate metrics}

A pseudo-Riemannian metric $h$ on $E$ is called {\em vertically
non-degenerate with respect to $\pi$} if its restriction $h_V$ to $V$
is non-degenerate. Then the bundle $H(h)$ of tangent vectors to $E$
which are $h$-orthogonal to $V$ is a complement of $V$ inside $TE$:
\be
TE=V \oplus H(h)
\ee
called the {\em horizontal distribution determined by $\pi$ and
$h$}. Moreover, $h$ restricts to a non-degenerate metric $h_H$ on
$H(h)$.

Let $h$ be a vertically non-degenerate metric with respect to
$\pi$ and let $H\eqdef H(h)$. Let $\nabla^V$ and $\nabla^H$
denote the vertical and horizontal covariant derivatives induced by the
Levi-Civita connection of $(E,h)$. Notice that $h_V\eqdef
h|_V\in \Gamma(E,\Sym^2(V^\ast))$ is a symmetric vertical
2-tensor field defined on $E$. 

\begin{definition}
The {\em vertical first fundamental form of a section $s\in
   \Gamma(\pi)$} is the symmetric $2$-tensor field $s^\ast_v(h_V)\in
 \Gamma(M,\Sym^2(T^\ast M))$:
\begin{equation*}
s^\ast_v(h_V)(X,Y)\eqdef h_V((\dd^v s)(X),(\dd^v s)(Y))\, , \quad \forall\, X,Y\in \cX(M)\, .
\end{equation*}
Given a metric $g$ on $M$, the {\em vertical second fundamental form}
of $s$ is the tensor field $\mathfrak{S}^v(s)\eqdef \nabla
\widehat{\dd^v s}$, where $\nabla$ is the connection induced on the
bundle $T^\ast M\otimes V^s$ by the Levi-Civita connection of $(M,g)$
and by the $s$-pullback of $\nabla^V$, while $\dd^v s$ is the vertical
differential of $s$ relative to the horizontal distribution $H(h)$.
\end{definition}

\begin{remark}
Notice that $\mathfrak{S}^v(s)$ need not be symmetric (even when $\pi$
is a pseudo-Riemannian submersion with totally geodesic fibers, see
below).
\end{remark}

\subsection{The Kaluza-Klein correspondence}

Recall the notion of pseudo-Riemannian submersion \cite[Chapter 7,
Definition 44]{ONeillBook}:

\begin{definition}
Let $\pi:E\rightarrow M$ be a surjective submersion and $h$, $g$ be
pseudo-Riemannian metrics on $E$ and $M$ respectively. We say that
$\pi$ is a {\em pseudo-Riemannian submersion from} $(E,h)$ to $(M,g)$
if the following two conditions are satisfied:
\begin{enumerate}[1.]
\item $h$ is vertically non-degenerate with respect to $\pi$, i.e. the
  restriction of $h$ to $V$ is non-degenerate.
\item The based isomorphism $(\dd\pi)_H:H\stackrel{\sim}{\rightarrow}
  (TM)^\pi$ is an isometry from $(H,h_H)$ to $((TM)^\pi,g^\pi)$.
\end{enumerate}
\end{definition}

\noindent Let $\pi:E\rightarrow M$ be a surjective submersion.  Given
a pseudo-Riemannian metric $g$ on $M$, the pseudo-Riemannian metrics
$h$ on $E$ which make $\pi:E\rightarrow M$ into a pseudo-Riemannian
submersion from $(E,h)$ to $(M,g)$ are parameterized by horizontal
distributions $H\in \cH(\pi)$ and by non-degenerate metrics on $V$, as
we explain below.

\begin{definition}
Let $g$ be a pseudo-Riemannian metric on $M$. A pseudo-Riemannian
metric $h$ on $E$ is called {\em compatible with $\pi$ and $g$} if
$\pi$ is a pseudo-Riemannian submersion from $(E,h)$ to $(M,g)$. It is
called {\em compatible with $\pi$} if there exists a pseudo-Riemannian
metric $g$ on $M$ such that $h$ is compatible with $\pi$ and $g$.
\end{definition}

\noindent 
Let $\Met_{\pi, g}(E)$ denote the set of all pseudo-Riemannian metrics
on $E$ which are compatible with $\pi$ and $g$ and $\Met_\pi(E)$
denote the set of all pseudo-Riemannian metrics on $E$ which are
compatible with $\pi$. For any $h\in \Met(E)$ which is vertically
non-degenerate, there exists at most one $g\in \Met(M)$ such that $h$
is compatible with $\pi$ and $g$. Indeed, $g$ is determined by $\pi$
and $h$ through the relation:
\be
g(X,Y)=h(\overline{X},\overline{Y})~~,~~\forall X,Y\in \cX(M)~~,
\ee
where $\overline{X}$ and $\overline{Y}$ denote the horizontal lifts of
$X$ and $Y$ relative to $H(h)$. Thus
$\Met_\pi(E)=\sqcup_{g\in\Met(M)}\Met_{\pi,g}(E)$. Let $\Met(V)$
denote the set of pseudo-Riemannian metrics on $V$.

\begin{definition}
Let $g\in \Met(M)$ be a pseudo-Riemannian metric on $M$,
$\fh\in \Met(V)$ be a pseudo-Riemannian metric on $V$ and
$H\in \cH(\pi)$ be a horizontal distribution for $\pi$. The {\em
Kaluza-Klein metric} determined by $g$, $\fh$ and $H$ is the
pseudo-Riemannian metric on $E$ defined through:
\ben
\label{KK}
h_{g, \fh, H}(X,Y)\eqdef \fh(X_V,Y_V)+g^\pi( (\dd \pi)_H
(X_H),(\dd \pi)_H(Y_H))\, , \quad \forall\, X,Y\in \cX(E)\, .
\een
\end{definition}

\noindent 
Notice that $h_{g,\fh,H}|_{V}=\fh$, so the Kaluza-Klein metric is
vertically non-degenerate with respect to $\pi$. Moreover, we have
$h_{g,\fh,H}|_H=g^\pi\circ ((\dd\pi)_H\otimes (\dd\pi)_H)$, hence
$(\dd \pi)_H$ is an isometry between the pseudo-Euclidean
distributions $(H, h_{g,\fh,h}|_H)$ and $((TM)^\pi, g^\pi)$. Thus the
Kaluza-Klein metric is compatible with $\pi$ and $g$, i.e. we have
$h_{g,\fh,H}\in \Met_{\pi,g}(E)\subset \Met_\pi(E)$.

\begin{definition}
The {\em Kaluza-Klein correspondence} of $\pi$ is the map
$\mathcal{K}_\pi:\Met(M)\times \Met(V)\times\cH(\pi)\rightarrow
\Met_\pi(E)$ defined through:
\be
\mathcal{K}_\pi(g,\fh,H)\eqdef h_{g,\fh,H}\in \Met_{\pi,g}(E)\subset \Met_\pi(E)\, .
\ee
\end{definition}

\

\noindent
The proof of the following statement is obvious:

\begin{proposition}
Let $h$ be a pseudo-Riemannian metric on $E$ and $g$ be a
pseudo-Riemannian metric on $M$. Then the following statements are
equivalent:
\begin{enumerate}[1.]
\item $h$ is compatible with $\pi$ and $g$, i.e. we have $h\in
  \Met_{\pi,g}(E)$.
\item $h$ is vertically non-degenerate with respect to $\pi$ and
  coincides with the Kaluza-Klein metric determined by $H\eqdef H(h)$,
  $g$ and by the metric $h_V\eqdef h|_{V}\in \Met(V)$.
\end{enumerate}
In particular, $h$ is uniquely determined by $H(h)$, $g$ and $h_V$.
\end{proposition}

\begin{corollary}
For every $g\in \Met(M)$, the map $\mathcal{K}_\pi(g,\cdot,
\cdot):\Met(V)\times \cH(\pi)\rightarrow \Met_{\pi,g}(E)$ is
bijective. In particular, the Kaluza-Klein correspondence
$\mathcal{K}_\pi$ is a bijection from $\Met(M)\times\Met(V)\times
\cH(\pi)$ to $\Met_\pi(E)$.
\end{corollary}

\noindent Hence any $\pi$-compatible metric $h$ on $E$ determines a
unique metric $g$ on $M$ such that $h$ is compatible with $\pi$ and
$g$. Moreover, $h$ is the Kaluza-Klein metric determined
by $\pi$, $g$, a unique vertical metric $\fh\in \Met(V)$ and a unique
horizontal distribution $H$ for $\pi$, namely $\fh=h_V$ and $H=H(h)$.

If $g\in \Met(M)$ has signature $(p,q)$ and $\fh\in \Met(V)$ has
signature $(p_V,q_V)$, then the Kaluza-Klein metric
$h=\mathcal{K}_\pi(g,\fh,H)$ determined by $g$, $\fh$ and any
horizontal distribution $H$ has signature $(p+p_V,q+q_V)$. When $g$ is
Lorentzian, the Kaluza-Klein metric $h$ is Lorentzian iff $\fh$ is
positive-definite, i.e. iff $q_V=0$. Let $\dim M=d$ and $\dim E=d+n$
where $n>0$. Then $\rk V = n$. The restriction of the Kaluza-Klein
correspondence gives a bijection:
\be
\mathcal{K}_\pi:\Met_{d-1,1}(M)\times \Met_{n,0}(V)\times \cH(\pi)\stackrel{\sim}{\rightarrow} \Met^{d+n-1,1}_\pi(E)\, ,
\ee
where $\Met^{d+n-1,1}_\pi(g)\eqdef \Met_\pi(E)\cap \Met_{d+n-1,1}(E)$
is the set of all Lorentzian metrics on $E$ which are compatible with
$\pi$.

\begin{definition}
A {\em Lorentzian submersion} from $(E,h)$ to $(M,g)$ is a surjective
pseudo-Riemannian submersion from $(E,h)$ to $(M,g)$, where $(E,h)$
and $(M,g)$ are Lorentzian manifolds.
\end{definition}

\noindent 
The remarks above show that the set of Lorentzian metrics $h$ on $E$
for which $\pi:(E,h)\rightarrow (M,g)$ is a Lorentzian submersion is
in bijection with the set $\Met_{d,0}(V)\times \cH(\pi)$ through the
Kaluza-Klein correspondence.

\section{Local description in adapted coordinates}
\label{app:localadapted}

Let $\pi:(E,h)\rightarrow (M,g)$ be a Lorentzian Kaluza-Klein space. 
Let $(W,u^0,\ldots, u^3, y^1, \ldots, y^n)$ be a $\pi$-adapted
coordinate chart on $E$ with corresponding coordinate chart
$(U,x^0,\ldots, x^3)$ on the base, where $U=\pi(W)$ and
$u^\mu=x^\mu\circ \pi$. We use lowercase Greek letters for indices
corresponding to the base, uppercase Latin letters for indices related
to the fiber and lowercase Latin letters for indices related to the
total space $E$.

We have $(\dd \pi)(\frac{\partial}{\partial
u^\mu})=\frac{\partial}{\partial x^\mu}$ and
$(\dd\pi)(\frac{\partial}{\partial y^A})=0$. Thus
$(\frac{\partial}{\partial y^1}\ldots \frac{\partial}{\partial y^n})$
is a local frame of the vertical distribution $V$ defined on $W$. Any
local frame $e_0\ldots e_3$ of the horizontal distribution $H$ which
is defined on $W$ has the form:
\ben
\label{e}
e_\mu=e_\mu^{~\nu}\frac{\partial}{\partial u^\nu}+v_\mu^{~A} \frac{\partial}{\partial y^A}\, ,
\een
for some coefficient functions $e_\mu^{~\nu},
v_\mu^{~A}\in\cC^\infty(W,\R)$, where the determinant of the
matrix-valued function $e \eqdef (e_\mu^{~\nu})_{\mu,\nu=0\ldots
3}\neq 0$ does not vanish on $W$. Setting
$f:=(f_\mu^{~\nu})_{\mu,\nu=0\ldots 3}\eqdef
e^{-1}\in \cC^\infty(W,\mathrm{Mat}(4,\R))$, we have
$f_{\mu}^{~\nu}e_{\nu}^{~\rho}=e_{\mu}^{~\nu}f_\nu^{~\rho}=\delta_\mu^{~\rho}$
and:
\ben
\label{f}
\frac{\partial}{\partial u^\mu}=f_{\mu}^{~\nu}e_\nu-f_\mu^{~A} \frac{\partial}{\partial y^A}\, ,
\een
where we defined: 
\be
f_\mu^{~A}\eqdef f_{\mu}^{~\nu} v_\nu^{~A}~~.
\ee
We also have $(\dd \pi)(e_\mu)=e_\mu^{~\nu}\frac{\partial}{\partial
  x^\mu}$. Given a vector field on $M$ locally expanded as $X=X^\mu
\frac{\partial}{\partial x^\mu}$, its horizontal lift defined by $H$
is given by ${\bar X}= X^\nu f_\nu^{~\mu} e_\mu$, so $f_\mu^{~\nu}$
are the coefficients of the vertical lift morphism with respect to the
local frames $\frac{\partial}{\partial x^\mu}$ of $TX$ and $e_\mu$ of
$H$:
\be
\overline{\frac{\partial}{\partial x^\mu}}=f_\mu^{~\nu}e_\nu~~.
\ee 

\noindent	
Let: 
\be
h_{\mu\nu}\eqdef_W h(e_\mu,e_\nu)~~,~~h_{AB}\eqdef_W h(\frac{\partial}{\partial y^A}, \frac{\partial}{\partial y^B})~~,~~ g_{\mu\nu}\eqdef_U g(\frac{\partial}{\partial x^\mu}, \frac{\partial}{\partial x^\nu})~~,
\ee
where $h_{\mu\nu}$ and $h_{AB}$ are smooth functions defined on $W$
while $g_{\mu\nu}$ are smooth functions defined on $U$. Since $H$ and
$V$ are $h$-orthogonal, we have $h(e_\mu, \frac{\partial}{\partial
y^A})=0$. Since $\dd\pi$ restricts to an isometry between $H$ and
$(TM)^\pi$, we have:
\be
h_{\mu\nu}=g_{\rho\sigma}^\pi e_\mu^{~\rho}e_\nu^{~\sigma}~~,
\ee
where $g_{\mu\nu}^\pi \eqdef g_{\mu\nu}\circ \pi$. In the local
coordinates chosen, this means:
\be
h_{\mu\nu}(u,y)=h_{\mu\nu}(u) = h_{\mu\nu}(x\circ\pi) = h_{\mu\nu}(x)\, .
\ee
Setting:
\be
e^\mu\eqdef f_\rho^{~\mu}\dd u^\rho~~,~~e^A\eqdef \dd y^A-f_\mu^A \dd u^\mu~~,
\ee
we have:
\beqa
&& e^\mu(e_\nu)=\delta^\mu_\nu~~,~~e^\mu(\frac{\partial}{\partial y^B})=0 \\
&& e^A(e_\nu)=0~~,~~e^A(\frac{\partial}{\partial  y^B})=\delta^A_B~~,
\eeqa
which shows that $(e^\mu,e^A)$ is the coframe dual to the frame
$(e_\mu, \frac{\partial}{\partial y^A})$. The metric $h$ has the
well-known Kaluza-Klein form:
\ben
\label{hlocal}
h=_W h_{\mu\nu} e^\mu\odot e^\nu+h_{AB} e^A\odot e^B=g_{\mu\nu}^\pi 
\dd u^\mu \odot \dd u^\nu +h_{AB}(\dd y^A-f_\mu^A\dd u^\mu)\odot (\dd y^B-f_\mu^B\dd u^\mu) ~~,
\een
which depends only on $f_\mu^A$. Thus: 
\beqa
&& h_V=_W h_{AB} e^A\odot e^B=h_{AB}(\dd y^A-f_\mu^A\dd u^\mu)\odot (\dd y^B-f_\mu^B\dd u^\mu) \\ && \equiv h_{AB}(x,y)(\dd y^A-f_\mu^A(x,y)\dd x^\mu)\odot (\dd y^B-f_\mu^B(x,y)\dd x^\mu)\, ,\\
&& h_H=_W h_{\mu\nu} e^\mu\odot e^\nu=g_{\mu\nu}^\pi \dd u^\mu\odot \dd u^\nu\equiv g_{\mu\nu}(x) \dd x^\mu\odot \dd x^\nu~~,
\eeqa
where after the sign $\equiv$ we identified $u$ with $x$ by abuse of
notation. Here $h_{AB}$ and $f_\mu^A$ depend on both $x$ and $y$.

\subsection{Local expression of the vertical differential and of the vertical first fundamental form of a section}

Any section $s\in \Gamma(\pi)$ satisfies $\pi\circ s=\id_M$, which
implies $u^\mu\circ s=_U x^\mu$. Defining:
\be
\varphi^A \eqdef y^A \circ s|_U\in \cC^\infty(U,\R)~~,
\ee
the section has the local expression:
\be
u^\mu(s(x))=x^\mu~~,~~y^A=\varphi^A(x)
\ee
in the coordinate charts $(U,x)$ of $M$ and $(W,x,y)$ of $E$. Setting
$\partial_\mu \varphi^A\eqdef \frac{\partial\varphi^A}{\partial
x^\mu}$, we have:
\be
(\dd s)(\frac{\partial}{\partial x^\mu})=\frac{\partial}{\partial u^\mu}+ \partial_\mu\varphi^A\frac{\partial}{\partial y^A}= f_{\mu}^{~\nu}e_\nu+(\partial_\mu \varphi^A-f_\mu^{~A}) \frac{\partial}{\partial y^A}~~,
\ee
which gives the local expression of the vertical differential of $s$:
\be
(\dd^v s)(\frac{\partial}{\partial x^\mu})=(\partial_\mu\varphi^A-f_\mu^A)\frac{\partial}{\partial y^A}~~.
\ee
In turn, this gives the local expression of the vertical first fundamental form: 
\be
s^\ast_v(h_V)(\frac{\partial}{\partial x^\mu}, \frac{\partial}{\partial x^\nu})=
(h_{AB}\circ s) (\partial_\mu\varphi^A-f_\mu^A\circ s) (\partial_\nu \varphi^B -f_\nu^B\circ s)\, .
\ee

\subsection{Local form of the Lagrange density of the section sigma model}

Using the formulas obtained above, we find that the Lagrange density
of the section sigma model has the coordinate expression:
\begin{eqnarray}
\label{evlocal}
e^v_\bPhi(s)=_U g^{\mu\nu} (h_{AB}\circ s)(\partial_\mu\varphi^A-f_\mu^A\circ s)(\partial_\nu \varphi^B-f_\nu^B\circ s)+\bPhi(s)\, ,
\end{eqnarray}
i.e.:
\ben
\label{evlocalcoords}
e^v(s)(x)=_U g^{\mu\nu}(x)h_{AB}(x,\varphi(x)) [\partial_\mu\varphi^A(x)-f_\mu^A(x,\varphi(x))] [\partial_\nu \varphi^B-f_\nu^B(x,\varphi(x))]+\bPhi(x,\varphi(x))~~.
\een
Hence the action for a relatively compact subset $U_0\subset U$ takes the form: 
\begin{eqnarray*}
\label{Slocal}
&& S_{\sc,U_0}[s]=S_{\sc,U_0}[\varphi]=\int_{U_0}\dd^4 x\sqrt{|g|} 
\left\{ g^{\mu\nu}(x)h_{AB}(x,\varphi(x)) [\partial_\mu\varphi^A(x)-f_\mu^A(x,\varphi(x))][\partial_\nu \varphi^B-f_\nu^B(x,\varphi(x))]\right. \\  && \left.  + \bPhi(x,\varphi(x))\right\}\, ,
\end{eqnarray*}
where $|g|\eqdef |\det (g_{\mu\nu})|$.

\subsection{Local form of the vertical tension field}

Let $\Gamma_{ij}^k=\Gamma_{ji}^k$ denote the Levi-Civita coefficients
of $h$ in the coordinate chart $(W,x,y)$. Thus:
\beqa
&& \nabla_{\frac{\partial}{\partial u^\mu}}(\frac{\partial}{\partial u^\nu})=
\Gamma_{\mu\nu}^\rho \frac{\partial}{\partial u^\rho} + \Gamma^A_{\mu\nu} \frac{\partial}{\partial y^A}=
\Gamma_{\mu\nu}^\rho f_\rho^\sigma e_\sigma + (\Gamma^A_{\mu\nu}-\Gamma_{\mu\nu}^\rho f_\rho^A) \frac{\partial}{\partial y^A}
\nn\\
&& \nabla_{\frac{\partial}{\partial u^\mu}}(\frac{\partial}{\partial y^A})=\nabla_{\frac{\partial}{\partial y^A}}(\frac{\partial}{\partial u^\mu})=
\Gamma_{\mu A}^\nu \frac{\partial}{\partial u^\nu} + \Gamma^B_{\mu A} \frac{\partial}{\partial y^B}=
\Gamma_{\mu A}^\nu f_\nu^{~\rho}e_\rho + (\Gamma^B_{\mu A}-\Gamma_{\mu A}^\nu f_\nu^B) \frac{\partial}{\partial y^B}
\nn\\
&& \nabla_{\frac{\partial}{\partial y^A}}(\frac{\partial}{\partial y^B})=\Gamma_{AB}^\mu \frac{\partial}{\partial u^\mu} 
+ \Gamma^C_{AB} \frac{\partial}{\partial y^C}=\Gamma_{AB}^\mu f_{\mu}^{~\nu}e_\nu 
+ (\Gamma^C_{AB}-\Gamma_{AB}^\mu f_\mu^{C}) \frac{\partial}{\partial y^C}~~,
\eeqa
with the symmetry properties $\Gamma_{\mu\nu}^\rho=\Gamma_{\nu\mu}^\rho$, $\Gamma_{AB}^C=\Gamma_{BA}^C$ , $\Gamma_{AB}^\mu=\Gamma_{BA}^\mu$, $\Gamma_{A\mu}^\nu=\Gamma_{\mu A}^\nu$ and $\Gamma_{A\mu}^B=\Gamma_{\mu A}^B$. This gives:
\beqa
&& \nabla^v_{\frac{\partial}{\partial u^\mu}}(\frac{\partial}{\partial
	u^\nu})=(\Gamma^A_{\mu\nu}-\Gamma_{\mu\nu}^\rho f_\rho^A) \frac{\partial}{\partial y^A}\\
&& \nabla^v_{\frac{\partial}{\partial u^\mu}}(\frac{\partial}{\partial
	y^A})= \nabla^v_{\frac{\partial}{\partial
		y^A}}(\frac{\partial}{\partial
	u^\mu})=(\Gamma^B_{\mu A}-\Gamma_{\mu A}^\nu f_\nu^B)\frac{\partial}{\partial y^B} \\ 
&&\nabla^v_{\frac{\partial}{\partial
		y^A}}(\frac{\partial}{\partial
	y^B})=(\Gamma^C_{AB}-\Gamma_{AB}^\mu f_\mu^{C})\frac{\partial}{\partial y^C}~~. 
\eeqa
On the other hand, we have: 
\be
\nabla_{\frac{\partial}{\partial x^\mu}}(\frac{\partial}{\partial x^\nu})=K_{\mu\nu}^\rho \frac{\partial}{\partial x^\rho}\, .
\ee
Let $\partial_A f_\nu^C\eqdef \frac{\partial f_\nu^C}{\partial y^A}$. Using the relations above, we compute:
\beqa
&& \nabla^v_{(\dd s)(\frac{\partial}{\partial x^\mu})} (\dd^v s(\frac{\partial}{\partial x^\nu}))=\nabla^v_{\frac{\partial}{\partial u^\mu}+
	\partial_\mu\varphi^A\frac{\partial}{\partial y^A}}\left[(\partial_\nu\varphi^B-f_\nu^B)\frac{\partial}{\partial y^B}\right]=\nn\\
&&=\nabla^v_{\frac{\partial}{\partial u^\mu}}\left[(\partial_\nu\varphi^B-f_\nu^B)\frac{\partial}{\partial y^B}\right]+ \partial_\mu\varphi^A\nabla^v_{\frac{\partial}{\partial y^A}}\left[(\partial_\nu\varphi^B-f_\nu^B)\frac{\partial}{\partial y^B}\right]=\nn\\
&&=\left[\partial_\mu(\partial_\nu\varphi^C-f_\nu^C)-\partial_\mu\varphi^A\partial_Af_\nu^C+
(\Gamma_{\mu A}^C-\Gamma_{\mu A}^\rho f_\rho^C)(\partial_\nu\varphi^A-f_\nu^A)+
(\Gamma^C_{AB}-\Gamma^\rho_{AB}f_\rho^C)\partial_\mu\varphi^A
(\partial_\nu\varphi^B-f_\nu^B)\right]\frac{\partial}{\partial y^C}~~
\eeqa
and:
\be
\dd^v s(\nabla_{\frac{\partial}{\partial x^\mu}}(\frac{\partial}{\partial x^\nu}))= K_{\mu\nu}^\rho (\partial_\rho\varphi^C-f_\rho^C)\frac{\partial}{\partial y^C}~~.
\ee
Hence the vertical tension filed of $s$ is given by:
\be
\tau^v(s)= g^{\mu\nu}\left[\nabla^v_{(\dd s)(\frac{\partial}{\partial x^\mu})} (\dd^v s(\frac{\partial}{\partial x^\nu}))-
\dd^v s(\nabla_{\frac{\partial}{\partial x^\mu}}(\frac{\partial}{\partial x^\nu}))\right]=\tau^v(s)^C\frac{\partial}{\partial y^C}~~,
\ee
where:
\begin{eqnarray*}
\tau^v(s)^C= g^{\mu \nu}\left[ \partial_\mu(\partial_\nu\varphi^C-f_\nu^C)-\partial_\mu\varphi^A\partial_Af_\nu^C+ (\Gamma_{\mu A}^C-\Gamma_{\mu A}^\rho f_\rho^C)(\partial_\nu\varphi^A-f_\nu^A)+\right. \\  \left. (\Gamma^C_{AB}-\Gamma^\rho_{AB}f_\rho^C)\partial_\mu\varphi^A	(\partial_\nu\varphi^B-f_\nu^B)-K_{\mu\nu}^\rho (\partial_\rho\varphi^C-f_\rho^C)
\right]\, .
\end{eqnarray*}
This can also be written as:
\beqan
\label{tauvlocal}
&&\tau^v(s)^C=\partial^\mu\partial_\mu \varphi^C + \nn\\
&& + ~(\Gamma^C_{AB}-\Gamma^\rho_{AB}f_\rho^C)\partial_\mu\varphi^A\partial^\mu\varphi^B  +
[\Gamma_{\mu A}^C-\Gamma_{\mu A}^\rho f_\rho^C -\partial_A f_\mu^C+ (\Gamma^\rho_{AB}f_\rho^C-\Gamma_{AB}^C)f_\mu^B] \partial^\mu\varphi^A \\
&&\nn - g^{\mu\nu}K_{\mu\nu}^\rho \partial_\rho\varphi^C + ~g^{\mu\nu}K_{\mu\nu}^\rho f_{\rho}^C+ g^{\mu\nu}(\Gamma_{\mu A}^\rho f_\rho^C-\Gamma_{\mu A}^C)f_\nu^A  -\partial^\mu f_\mu^C\, .
\eeqan
In particular, the harmonic section equation for $s$ amounts to a
system of second order PDEs for the functions $\varphi^C$, which
contain terms of order zero and can be viewed as deformations of the
d'Alembert equation.

\subsection{Local form of the curvature of $H$}

The Lie bracket of the vector fields given in \eqref{e} has the form:
\beqa
&& [e_\mu,e_\nu]=(e_\mu^{~\rho}\partial_\rho
e_\nu^{~\sigma}-e_\nu^{~\rho}\partial_\rho e_\mu^{~\sigma}+ e_\mu^A \partial_A
e_\nu^{~\sigma}-e_\nu^A\partial_A e_\mu^{~\sigma})\frac{\partial}{\partial
	u^\sigma}\nn \\ && + (e_\mu^{~\rho} \partial_\rho e_\nu^C -e_\nu^{~\rho} \partial_\rho e_\mu^C + e_\mu^A\partial_A e_\nu^B- e_\nu^A\partial_A e_\mu^B)\frac{\partial}{\partial y^C}=\nn\\
&& =(e_\mu^\rho\partial_\rho
e_\nu^\sigma-e_\nu^{~\rho}\partial_\rho e_\mu^{~\sigma}+ e_\mu^A \partial_A
e_\nu^{~\sigma}-e_\nu^A\partial_A e_\mu^{~\sigma})f_\sigma^{~\tau} e_\tau+ \nn\\
&& [e_\mu^{~\rho} \partial_\rho e_\nu^C -e_\nu^{~\rho}
\partial_\rho e_\mu^C + e_\mu^A\partial_A e_\nu^C- e_\nu^A\partial_A e_\mu^C
-f_\sigma^C(e_\mu^{~\rho}\partial_\rho
e_\nu^{~\sigma}-e_\nu^{~\rho}\partial_\rho e_\mu^{~\sigma}+ e_\mu^A \partial_A
e_\nu^{~\sigma}-e_\nu^A\partial_A e_\mu^{~\sigma}) ]\frac{\partial}{\partial y^C}~~.
\eeqa
Hence:
\be
\cF(e_\mu,e_\nu)=\cF_{\mu\nu}^C\frac{\partial}{\partial y^C}~~,
\ee
where: 
\be
\cF_{\mu\nu}^C=e_\mu^{~\rho} \partial_\rho e_\nu^C -e_\nu^{~\rho} \partial_\rho
e_\mu^C + e_\mu^A\partial_A e_\nu^C- e_\nu^A\partial_A
e_\mu^C-f_\sigma^C(e_\mu^{~\rho}\partial_\rho
e_\nu^{~\sigma}-e_\nu^{~\rho}\partial_\rho e_\mu^{~\sigma}+ e_\mu^A
\partial_A e_\nu^{~\sigma}-e_\nu^A\partial_A e_\mu^{~\sigma})~~.
\ee
Thus $H$ is integrable iff $\cF_{\mu\nu}^C=0$. The previous
calculation uses a local frame $\left\{ e_{\mu}\right\}_{\mu =
  0,\hdots,3}$ of $H$. Any other local frame of $H$ defined above $W$
has the form:
\be
e'_\mu=a_\mu^{~\nu}e_\nu~~,
\ee
where $a=(a_\mu^{~\nu})_{\mu,\nu=0\ldots 3}$ is an invertible
matrix. We have $e'=ae$ and $(e')_\mu^A=a_\mu^{~\nu}e_\nu^{~A}$, thus
$f'=fa^{-1}$ and $(f')_\mu^{~A}=f_\mu^{~A}$. The metric \eqref{hlocal}
is invariant under such changes of the local frame of $H$. Since $e$
is non-degenerate, by taking $a=e^{-1}$ we can always insure that
$e'=I_4$.  Hence $H$ always admits a local frame with
$e_\mu^{~\nu}=\delta_\mu^\nu$, which we call a {\em special
  frame}. For such a frame of $H$, we have
$f_\mu^{~\nu}=\delta_\mu^\nu$, $f_\mu^A=v_\mu^A$ and:
\be
e_\mu=\frac{\partial}{\partial u^\mu}+v_\mu^{~A} \frac{\partial}{\partial y^A}=
\frac{\partial}{\partial u^\mu}+f_\mu^{~A} \frac{\partial}{\partial y^A}~~.
\ee
When expressed using a special frame of $H$, the curvature coefficients become: 
\ben
\cF_{\mu\nu}^C=\partial_\mu f_\nu^C-\partial_\nu f_\mu^C+f_\mu^A\partial_A f_\nu^C-f_\nu^A\partial_A f_\mu^C~~.
\een

\subsection{The case when $H$ is integrable}
\label{subsec:Hintlocal}
A simple particular case arises when $H|_W$ is Frobenius
integrable. In this situation, one can find adapted coordinates
$(W,u,y)$ such that $H$ identifies locally with the integrable
distribution $H=(TM)^\pi$ (which is spanned by
$\frac{\partial}{\partial u^0}\ldots \frac{\partial}{\partial u^3}$)
and such that $\frac{\partial}{\partial u^\mu}$ are invariant under
the Ehresmann transport of $H$.  We can then take
$e_\mu=\frac{\partial}{\partial u^\mu}$, which corresponds to
$e_\mu^{~\nu}=\delta_\mu^\nu$ and $e_\mu^{~A}=0$. Thus $f_\mu^A=0$ and
the Kaluza-Klein metric reduces to:
\ben
\label{h0}
h=_W g_{\mu\nu}^\pi \dd u^\mu\odot \dd u^\nu+h_{AB} \dd y^A \odot \dd y^B~~.
\een
Since $h_{V}$ is invariant under Ehresmann transport (which proceeds
through isometries of the fibers) the coefficients $h_{AB}$ are
independent of $x$ and coincide with the metric coefficients of the
fiber:
\ben
\label{hflat}
h_{AB}=\cG_{AB}(y)~~.
\een
Hence the metric \eqref{h0} reduces to the product metric $g\times
\cG$. We have $\Gamma_{\mu A}^C=0$, hence the coefficients
\eqref{tauvlocal} of the vertical tension field reduce to the local
form of an ordinary sigma model:
\be
\tau^v(s)^C=\partial^\mu\partial_\mu \varphi^C+\Gamma^C_{AB}\partial_\mu\varphi^A\partial^\mu\varphi^B - 
g^{\mu\nu}K_{\mu\nu}^\rho \partial_\rho\varphi^C
\ee
On the other hand, the restriction of $\bPhi$ to $W$ is a function:
\begin{equation}
\bPhi|_{W}\colon W\to \R\, ,
\end{equation}
which in principle depends both on $u^{\mu}$ and $y^{A}$. 
Equation \eqref{Vvert} implies:
\begin{equation}
\dd\bPhi|_{W} = \frac{\partial}{\partial u^{\mu}} \bPhi|_{W}\, \dd u^{\mu} + 
\frac{\partial}{\partial y^{A}} \bPhi|_{W}\, \dd y^{A}= \left( \frac{\partial}{\partial y^{A}} \bPhi|_{W} - f^{~A}_{\nu}\frac{\partial}{\partial y^{A}} \bPhi|_{W}\right)\, \dd y^A\, .
\end{equation}  
Since integrability of $H|_W$ implies $f^{A}_{\mu} = 0$ (see above), we
find that $\bPhi|_W$ only depends on $y^{A}$:
\begin{equation}
\frac{\partial}{\partial u^{\mu}} \bPhi|_{W} = 0\, , \quad \dd\bPhi|_{W} = \frac{\partial}{\partial y^{A}} \bPhi|_{W}\, \dd y^A\, . 
\end{equation}
Thus $\bPhi|_W(x,y)=\Phi_W(y)$ for some function $\Phi_W$. This shows
that the section scalar sigma model is locally equivalent with the
ordinary scalar sigma model when $H$ is integrable. A more geometric
explanation of this fact is given in Section
\ref{sec:scalarsection}.

\subsection{Local expression of the fundamental bundle form $\bTheta$ and of the fundamental bundle field $\bPsi$}

Let $\bcD$ be a scalar-electromagnetic bundle of type $\bcD =
(\pi,\bPhi,\bXi)$ with associated Kaluza-Klein space $\pi\colon
(E,h)\to (M,g)$. Let $\bDelta = (\bcS,\bomega,\bD)$ be the
corresponding duality bundle, namely $\bcS$ is a real vector bundle of
rank $2r$ over $E$ endowed with a symplectic structure $\bomega$ and a
compatible flat connection $\bD$. As explained in Section
\ref{sec:scalarelectro}, a taming $\bJ$ of $(\bcS,\bomega)$ is said to
be \emph{vertical} if it satisfies $\bD_{X}\circ \bJ = \bJ\circ
\bD_{X}$ for all $X\in\Gamma(E,H)$.  Let $(W,u^0,\ldots, u^3, y^1,
\ldots, y^n)$ be a $\pi$-adapted coordinate chart on $E$, with
corresponding coordinate chart $(U,x^0,\ldots, x^3)$ on $M$, where
$U=\pi(W)\subset M$ and $u^\mu=x^\mu\circ \pi$. Let $\cE := (
\epsilon_{M})_{M=1,\ldots, 2r}$ be a local flat symplectic frame of
$\bcS$ over $W\subset E$ and let $\cE^{\ast}
:=(\epsilon^{M})_{M=1,\ldots, 2r}$ denote the dual frame of
$\bcS^\ast$. In terms of $\cE$ and $\cE^{\ast}$, the restriction of
$\bJ$ to $W$ can be written as follows:
\ben
\label{Jexp}
\bJ|_{W} = \bJ_{M}^{\,\, N} \epsilon_{N}\otimes \epsilon^{M}\, ,
\een
where $\bJ_{M}^{\,\, N}\in C^{\infty}(W,\mathbb{R})$. Let
$H=h(h)\subset TE$ be the horizontal distribution of the Kaluza-Klein
space $\bcD$.  As already explained, $H$ is locally spanned on $W$ by
the tangent vector fields $e_\mu$ of equation \eqref{e}.  Since $\cE$
is a local flat frame, we have:
\begin{equation*}
\bD(\epsilon_{M}) = 0~~.
\end{equation*}
Using this relation, the verticality condition for $J$ amounts to
\begin{equation*}
\bD_{e_{\mu}}\bJ|_W = 0~~,
\end{equation*}
which in turn is equivalent with:
\ben
\label{Jvertlocal}
\epsilon_{\mu}^{\,\,\nu}\frac{\partial}{\partial u^{\nu}} \bJ_{M}^{\,\,N} + v_{\mu}^{A}\frac{\partial}{\partial y^{A}} \bJ_{M}^{\,\,N} = 0~~,
\een
where we used \eqref{Jexp} and \eqref{e}. 

Assume now that $H$ is integrable. Without loss of
generality, we can than take $e_{\mu}^{\,\,\nu} = \delta^{\nu}_{\mu}$ and
$v_{\mu}^{\,\, A} =0$, thus $e_\mu=\frac{\partial}{\partial u^\mu}$
(see Subsection \ref{subsec:Hintlocal}). In this case, the verticality
condition \eqref{Jvertlocal} for $\bJ$ reduces to:
\ben
\label{Jvertlocalint}
\frac{\partial}{\partial u^{\mu}} \bJ_{M}^{\,\,N} = 0\, , \qquad \mu = 0,\hdots, 3\, .
\een
This shows that $ \bJ_{M}^{\,\,N}$ depend only on $y^{A}$. Using the definition of the fundamental bundle
form $\bTheta$, we obtain:
\begin{equation*}
\bTheta|_{W}(e_{M}) = (\bD^{\ad}\bJ)|_{W} (e_{M}) = \bD (\bJ|_{W}(e_{M})) = \left(\frac{\partial}{\partial y^{A}} \bJ_{M}^{\,\, N}\right) \dd y^{A} \otimes e_{N}\, .
\end{equation*}
Similarly, the fundamental bundle field $\bPsi$ is given by:
\begin{equation*}
\bPsi|_{W}(e_{M}) = h^{AB} \left(\frac{\partial}{\partial y^{A}} \bJ_{M}^{\,\, N}\right) \frac{\partial}{\partial y^{B}} \otimes e_{N}~~.
\end{equation*}
Let $s$ be a section of $\pi$. The pull-backs by $s$ of $\bTheta$ and $\bPsi$ are given by:
\begin{equation}
\label{eq:localThetaPsi}
\bTheta^{s}|_{U}(e^{s}_{M}) = \left(\frac{\partial \bJ_{M}^{\,\, N}}{\partial y^{A}}\right)\circ \varphi \, \dd\varphi^{A} \otimes e^{s}_{N}\, , 
\quad \bPsi^{s}|_{U_{\alpha}}(e^{s}_{M}) = \cG^{AB}\circ \varphi \left(\frac{\partial \bJ_{M}^{\,\, N}}{\partial y^{A}} \right)\circ \varphi \, 
\left(\frac{\partial}{\partial y^{B}}\right)^s \otimes e^{s}_{N}\, ,
\end{equation}
where we used equation \eqref{hflat} and we defined:
\begin{equation*}
\varphi^{A} = y^{A}\circ s|_{U} \in \cC^{\infty}(U,\mathbb{R})~~.
\end{equation*}
Notice that the pullback of $\bJ$ by $s$ has the local form:
\begin{equation*}
\bJ^{s}|_{U} = \bJ_{M}^{\,\,N}\circ \varphi \, e^{s}_{N}\otimes e^{s M}~~.
\end{equation*}
Let $\cD = (\cS, \omega, D, J)$ be the type of $\bcD$ and consider a
special trivializing atlas $(U_{\alpha}, \mq_{\alpha})_{\alpha\in I}$
for the scalar-electromagnetic bundle $\bcD$. This induces local
isometric trivializations:
\begin{equation*}
q_{\alpha}\colon E_{\alpha} \to U_{\alpha}\times \cM
\end{equation*}
of $\pi\colon (E,h)\to (M,g)$ and unbased isomorphisms of electromagnetic structures:
\begin{equation*}
\mq_{\alpha}\colon (\bcS_{\alpha},\bomega_{\alpha},\bD_{\alpha},\bJ_{\alpha})\stackrel{\sim}{\rightarrow}  (\bcS^{0}_{\alpha},\bomega^{0}_{\alpha},\bD^{0}_{\alpha},\bJ^{0}_{\alpha})\, ,
\end{equation*}
where $\bcS_{\alpha}\eqdef \bcS|_{U_\alpha}$ etc. and:
\begin{equation*}
\bcS^{0}_{\alpha} = \cS^{p^{0}_{\alpha}}\, , \quad \bomega^{0}_{\alpha} = \omega^{p^{0}_{\alpha}}\, , \quad \bD^{0}_{\alpha} = D^{p^{0}_{\alpha}}\, , \quad \bJ^{0}_{\alpha} = J^{p^{0}_{\alpha}}~~.
\end{equation*}
In turn, $\mq_\alpha$ induce isomorphisms:
\begin{equation*}
\mq_\alpha^{(s)}:(\bcS_\alpha^s, \bomega_\alpha^s, \bD_\alpha^s,
\bJ_\alpha^s)\stackrel{\sim}{\rightarrow} (\cS^{\varphi^{\alpha}},
\omega^{\varphi^{\alpha}}, D^{\varphi^{\alpha}},
J^{\varphi^{\alpha}})\, .
\end{equation*}
Setting $U=U_\alpha$ in equation \eqref{eq:localThetaPsi} gives:
\begin{eqnarray}
\label{eq:localThetaPsiII}
\Theta^{\alpha}(e^{\alpha}_{M}) &=& \left(\frac{\partial}{\partial \varphi^{A}} J_{M}^{\,\, N}(\varphi^{\alpha})\right) \dd\varphi^{A} \otimes e^{\alpha}_{N}~~, \nonumber\\ 
\Psi^{\alpha}(e^{\alpha}_{M}) &=& \cG^{AB}(\varphi^{\alpha}) \left(\frac{\partial}{\partial \varphi^{A}} J_{M}^{\,\, N}(\varphi^{\alpha})\right) \frac{\partial}{\partial \varphi^{B}} \otimes e^{\alpha}_{N}~~.
\end{eqnarray}
Here $\Theta^{\alpha} \eqdef (\varphi^{\alpha})^{\ast} \Theta$ and
$\Psi^{\alpha} \eqdef (\varphi^{\alpha})^{\ast} \Psi$ (where $\Phi$
and $\Psi$ are the fundamental form and fundamental vector field of
$\cD$ defined in \cite{GESM}), while $e^{\alpha}_{M} =
\mq^{s}_{\alpha}(e^{s}_{M})$ form a flat symplectic frame $\cE^\alpha$
of $\cS^{\varphi^\alpha}$. We conclude that the local form of
$\bTheta^{s}$ and $\bPsi^{s}$ is consistent with that of generalized
Einstein-Scalar-Maxwell theories. We can identify $\bcV|_{U_\alpha}$
with the bundle-valued form:
\begin{equation*}
\cV^{\alpha} = \mq^{(s)}_{\alpha}(\bcV_{\alpha}) = F^{\Lambda} e^{\alpha}_{\Lambda} + G_{\Lambda} e^{\alpha, \Lambda} \in
\Omega^{2}_{\cl}(U_{\alpha},S^{\varphi^{\alpha}})~~,
\end{equation*}
where $F^{\Lambda}, G_{\Lambda} \in \Omega^{2}_{\cl}(U_{\alpha})$ for
$\Lambda = 1,\hdots,r$. In the symplectic frame $\cE^{\alpha}$, the
form $\varphi^{\alpha}(\omega)\in \Omega^2(U_\alpha)$ has matrix
components $\omega^{\alpha}_{MN}$ ($M, N =1,\hdots, 2r$) corresponding
to the standard symplectic matrix. The previous formulas imply:
\begin{eqnarray*}
(\ast \bcV , \bPsi^s\bcV) &=_{U_{\alpha}}&
(\ast\cV^{N}, \cV^{P})_{g}\, \omega^{\alpha}_{NM}
\left(\frac{\partial}{\partial \varphi^{A}} J_{P}^{\,\,
  M}(\varphi)\right) = -2 (\partial_{A} G_{\Lambda}, \ast
F^{\Lambda})_{g}\, ,\\ \bcV\loslash \bcV &=_{U_{\alpha}}&
(\omega^{\alpha}_{NM} J_{P}^{\,\, M})\, \cV^{N}_{\mu\rho}
\cV^{P}_{\delta\nu} g^{\rho\delta} \dd x^{\mu}\otimes \dd x^{\nu}\, ,
\end{eqnarray*}
as required by local consistency with standard ESM
theories. It is now easy to see that the
positive polarization condition:
\begin{equation*}
\ast_g \bcV=-\bJ^s\bcV\, ,
\end{equation*}
takes the following form when restricted to $U_{\alpha}$:
\begin{equation}
\label{eq:twistedself}
\ast_{g}\cV = -J^{\varphi^{\alpha}} \cV~~,
\end{equation}
thus reducing to the local twisted self-duality condition 
by of an ordinary ESM theory defined on $U_\alpha$.

\begin{acknowledgements}
The work of C. I. L. is supported by grant IBS-R003-S1. The work of
C. S. S. is supported by a Humboldt Research Fellowship from the
Alexander Von Humboldt Foundation. C.S.S. wishes to thank Vicente
Cort\'es for enlightening discussions and for pointing out that the
geodesic completeness condition is not necessary for the construction
presented in the paper.
\end{acknowledgements}

\end{document}